\documentclass[a4paper,11pt]{amsart}
\usepackage{amsfonts,amssymb,epsfig,latexsym}
\usepackage{amsmath}
\usepackage[latin1]{inputenc}
\usepackage{color,layout}
\usepackage{ifthen}

\usepackage{graphicx}
\usepackage{color}

\newtheorem{theorem}{Theorem}[section]

\newtheorem{proposition}[theorem]{Proposition}

\newtheorem{remark}[theorem]{Remark}

\def\square{\hbox{\vrule\vbox{\hrule\phantom{o}\hrule}\vrule}}

\def\eps{\varepsilon}

\newcommand{\be}{\begin{equation}}
\newcommand{\ee}{\end{equation}}

\numberwithin{equation}{section}

\newcommand{\Z}{\mathbb{Z}}
\newcommand{\R}{\mathbb{R}}
\newcommand{\C}{\mathbb{C}}
\newcommand{\W}{{\mathcal W}}

\newcommand{\cO}{{\mathcal O}}

\newcommand{\lan}{\langle}
\newcommand{\ran}{\rangle}
\newcommand{\pal}{\parallel}

\newcommand{\e}{\varepsilon}

\newcommand{\re}{{\rm Re}}
\newcommand{\im}{{\rm Im}}
\newcommand{\ord}{{\mathcal O}}
\newcommand{\ai}{{\rm Ai}\,}
\newcommand{\bi}{{\rm Bi}\,}
\newcommand{\K}{{\mathcal K}\,}

\def\eq#1{(\ref{#1})}

\newtheorem{thm}{Theorem}[section]

\newtheorem{rem}{Remark}[section]

\numberwithin{equation}{section}

\begin{document}

\title[Resonances at energy-level crossing]{Molecular predissociation resonances\\  at an energy-level crossing}
\author{S.~Fujii\'e${}^1$, A.~Martinez${}^2$ and T.~Watanabe${}^3$}



\maketitle
\addtocounter{footnote}{1}
\footnotetext{{\tt\small  Department of Mathematical Sciences, Ritsumeikan University, 111 Noji-Higashi, Kusatsu, 525-8577,  Japan, 
fujiie@fc.ritsumei.ac.jp} }  
\addtocounter{footnote}{1}
\footnotetext{{\tt\small Universit\`a di Bologna,  
Dipartimento di Matematica, Piazza di Porta San Donato, 40127
Bologna, Italy, 
andre.martinez@unibo.it }}  
\footnotetext{{\tt\small  Department of Mathematical Sciences, Ritsumeikan University,  111 Noji-Higashi, Kusatsu, 525-8577,  Japan,
t-watana@se.ritsumei.ac.jp} }  

\begin{abstract}
We study the resonances of a two-by-two semiclassical system of one dimensional Schr\"odinger operators, near an energy where the two potentials intersect transversally, one of them being bonding, and the other one anti-bonding. Under an ellipticity condition on the interaction, we obtain optimal estimates on the location and on the widths of these resonances. 
\end{abstract}
\pagebreak

\section{Introduction}
This paper is devoted to the study of diatomic molecular predissociation resonances in the Born-Oppenheimer approximation, at energies close to that of the crossing of the electronic levels.  In such a situation, we aim to provide precise estimates both on the real part and on the imaginary part (width) of the resonances. As it is well known, they respectively correspond to the radiation frequency and to the inverse of the life-time of the molecule.

In all of the work, the parameter $h$ stands for the square-root of the inverse of the (mean-) mass of the nuclei. The Born-Oppenheimer approximation permits to reduce the study to that of a semiclassical system of Schr\"odinger-type operators (see, e.g., \cite{KMSW, MaMe, MaSo}), and the size of the system depends on the number of electronic levels that are involved. For instance, at sufficiently low energies, this system is scalar, and, typically, one can apply the numerous results coming from the semiclassical analysis of the Schr\"odinger operator (see, e.e., \cite{Ha, HeSj, Ma, ReSi, HiSi, DiSj, DyZw, Zw} and references therein).

On the contrary, when several electronic levels are involved, only few results are available. One may quote \cite{Ba, Na, FLN, GrMa1}, where very particular potentials are considered, and \cite{Kl, GrMa2, MaBr}, where the potentials are much more general, but the energy considered is lower than that of the crossing. Actually, in this last situation the width of the resonances can be estimated by a tunnelling effect through a potential barrier, and it is exponentially small (in the parameter $h$). However, according to chemists, these widths correspond to such long life-times that it may even surpass the age of the universe! For that reason, it seems more reasonable (and, in any case, of interest) to consider situations where the widths are not that small. 

As a matter of fact, this is what is expected when the energy considered becomes very close to that of the crossing. However, this case is difficult to treat in general, because it corresponds to a somehow degenerate situation where the characteristic set of the operator has a singularity of the type $\{ |\xi|^4-x_1^2=0\}$.

Here, we study a model with one degree of freedom, where such a phenomenon occurs. Namely, we consider a $2\times2$ matrix system, the diagonal part of which consists of semiclassical Schr\"odinger operators, and the off-diagonal part of a lower order differential interaction. We assume that the two potentials cross transversally at the origin, with value 0, and that, at this energy level, one of the two potentials admits a well, while the other one is non-trapping (see figure 1).

\begin{figure}[h]
\label{fig1}
\begin{center}
\scalebox{0.5}[0.35]{
\includegraphics{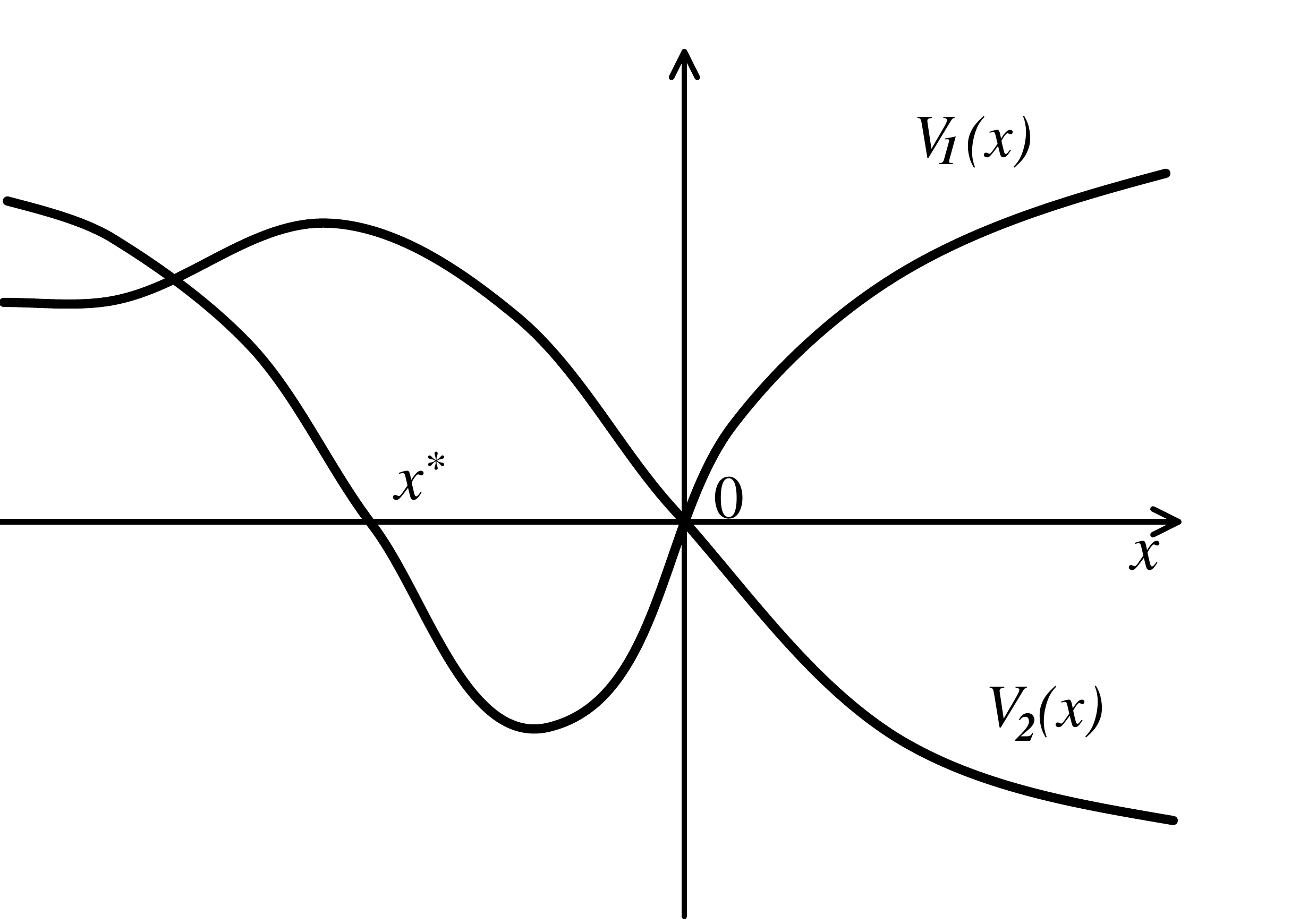}
}
\end{center}
\caption{The two potentials}
\end{figure}

For such a model, we study the resonances $E=E(h)$ that have a real part $\ord (h^{2/3})$ and an imaginary part $\ord (h)$. We actually prove their existence, and give their asymptotic behaviour as $h\to 0_+$. In particular, we find that their widths behave exactly like $h^{5/3}$, except possibly for particular values of the limit of $h^{-2/3}E(h)$, corresponding to positive zeros of some Airy-type function, and for which the width becomes $o(h^{5/3})$. 

It is important to notice that, for such a kind of system, none of the standard methods used in the scalar case can be applied. For instance, the exact WKB method (see, e.g., \cite{FuRa, Vo}) does not work, because of the presence of 4 phase-functions $\pm\varphi_j$, $j=1,2$ (two for each potentials). Indeed, the method requires that, for each of them, there exists at least one direction where the real part of difference with the three other ones increases, and here it cannot be the case. Even the formal semiclassical WKB constructions can be performed only at those points where the two potentials do not cross. Finally, the recent method proposed by D. Yafaev in \cite{Ya} (and from which this work has been mostly inspired) does not seem to be adaptable to our case.

Therefore, instead of trying to generalise the scalar methods, we have chosen to use particular fundamental solutions of the two scalar underlying Schr\"odinger operators, and to take advantage of the fact that everything is known on their behaviours (both semiclassical and at infinity) in order to solve the system by an iterative procedure. In that way, we can construct two exponentially decaying solutions on one side, and two outgoing solutions on the other side, with good estimates on their behaviour up to the interaction point where the two potentials cross. This makes possible to compute the Wronskian of these four solutions at that point, and obtain in this way the condition of quantization that determines the resonances. Then, a careful analysis of this condition leads to precise estimates on both the real part and the width of these resonances (see Theorem \ref{mainth}).

Let us mention that our methods still work for problems on $L^2(\R_+)$, with potentials $V_j(r)$ behaving like $c_j/r^\alpha$ in 0, with $c_j>0$ and $0\leq \alpha\leq 2$ (see Remark \ref{R+}). Moreover, under an additional assumption on $W$, perturbations of size $h^2$ can be admitted, too, such as semiclassical pseudodifferential operators of order $-2$, as it occurs after the Feshbach reduction in the Born-Oppenheimer reduction process (\cite{MaMe, KMSW, MaSo}): see Remark \ref{rempseudo}.

In the next section, we give the precise assumptions we work with, and we state our main result. Then, in Section \ref{sect3}, we construct fundamental solutions to the scalar operators both on the real negative half-line and on on the positive half-line, and give estimates on them. These fundamental solutions are then used, in Section \ref{sect4}, to construct  bounded solutions to the system in an iterative way, both on the negative half-line and on the positive half-line. In Section \ref{sect5}, these solutions and their derivatives are estimated more precisely at the crossing point, and, in Section \ref{sect6}, their Wronskian is computed and the quantization condition is written. In order to complete the proof, an analysis of this quantization condition is performed in Section \ref{sect7}, and additional informations are given concerning the resonant states. Finally, an appendix (Section \ref{sect8}) contains recalls on the Airy functions, some extensions of Yafaev's constructions for the scalar Schr\"odinger equation, and a list of formulas that may help the reader.

\section{Assumptions and results}

We consider a system of Schr\"odinger equations:
\begin{equation}
\label{sch}
\begin{aligned}
Pu &= Eu,\qquad
P &= \left(
\begin{matrix}
 P_1 & hW\\
hW^* & P_2
\end{matrix}
\right),
\end{aligned}
\end{equation}
where $D_x$ stands for $-i\frac{d}{d x}$, $P_j=h^2D_x^2 + V_j(x)$ ($j=1,2$), and $W=W(x,hD_x)$.

We suppose the following conditions on the potentials $V_1(x), 
V_2(x)$ (see Figure 1) and on the interaction
$W(x,hD_x)$:

{\bf Assumption (A1)}
$V_1(x)$, $V_2(x)$ 
 are real-valued analytic functions on $\R$,s and extend to holomorphic functions
 in the complex domain,
$$
\Gamma=\{x\in\C\,;\,|\im\, x|<\delta_0\lan\re \,x\ran\}
$$
where $\delta_0>0$ is a constant, and $\lan t\ran:=(1+|t|^2)^{1/2}$.

{\bf Assumption (A2)} For $j=1,2$, $V_j$ admits limits as $\re\, x\to \pm\infty$ in $\Gamma$, and they satisfy,
$$
\begin{aligned}
\lim_{{\re\,x\to -\infty}\atop{x\in \Gamma}}  V_1(x)>0\, ;\, \lim_{{\re\,x\to -\infty}\atop{x\in \Gamma}} V_2(x)>0\, ;\\
\lim_{{\re\,x\to +\infty}\atop{x\in \Gamma}} V_1(x)>0\, ;\, \lim_{{\re\,x\to +\infty}\atop{x\in \Gamma}} V_2(x)<0.
\end{aligned}
$$

{\bf Assumption (A3)} There exists a negative number $x^*<0$ such that,
\begin{itemize}
\item $V_1>0$ and $V_2>0$ on $(-\infty, x^*)$;
\item $V_1<0<V_2$ on $(x^*,0)$;
\item $V_2<0<V_1$ on $(0,+\infty)$,
\end{itemize}
and one has,
$$
V_1'(x^*)=:-\tau_0  <0,\quad  V_1'(0)=:\tau_1>0,\qquad V_2'(0)=:-\tau_2<0.
$$

{\bf Assumption (A4)}
$W(x,hD_x)$ is a first order differential operator
$$
W(x,hD_x)=r_0(x)+r_1(x)hD_x,
$$
where $r_0(x)$ and $r_1(x)$ are analytic and bounded in $\Gamma$, and $r_0(0) \not= 0$.

\vskip 0.3cm
In this situation, in a neighbourhood of the energy 0, the spectrum of $P$ is essential only, and the resonances of $P$ can be defined, e.g., as the values $E\in\C$ such that the equation $Pu=Eu$ has a non trivial outgoing solution $u$, that is, a non identically vanishing solution such that, for some $\theta >0$ sufficiently small, the function $x\mapsto u(xe^{i\theta})$ is in $L^2(\R)\oplus L^2(\R)$ (see, e.g., \cite{AgCo, ReSi}). Equivalently, the resonances are the eigenvalues of the operator $P$ acting on $L^2(\R_\theta)\oplus L^2(\R_\theta)$, where $\R_\theta$ is a complex distorsion of $\R$ that coincides with $e^{i\theta}\R$ for $x\gg 1$ (see, e.g., \cite{HeMa}). We denote by ${\rm Res}(P)$ the set of these resonances.

For $E\in \C$ small enough, we define the action,
$$
{\mathcal A}(E):= \int_{x_1^*(E)}^{x_1(E)}\sqrt{ E-V_1(t)} \, dt,
$$
where $x_1^*(E)$ (respectively $x_1(E)$) is the unique solution of $V_1(x)=E$ close to $x^*$ (respectively close to 0), and it is well-known that, in this situation, ${\mathcal A}(E)$ is an analytic function of $E$ near 0.

We also fix $C_0>0$ arbitrarily large, and we plan to study the resonances of $P$ lying in the set ${\mathcal D}_h(C_0)$ given by,
\be
{\mathcal D}_h(C_0):= [-C_0h^{2/3}, C_0h^{2/3}]-i[0,C_0h].
\ee
For $h>0$ and $k\in\Z$, we set,
\be
\label{deflambdakh}
\lambda_k(h):=\frac{-2{\mathcal A}(0)+(2k+1)\pi h}{2{\mathcal A}'(0)h^{2/3}},
\ee
where
${\mathcal A}'(0)=\int_{x^*}^0\frac1{2\sqrt{|V_1(x)|}} dx>0$ is the first derivative of ${\mathcal A}(E)$ at 0.
Then, our result is,

\begin{thm}\sl
\label{mainth}
Under Assumptions (A1)-(A4), for $h>0$ small enough, one has,
$$
{\rm Res}\,(P)\cap {\mathcal D}_h(C_0)
=\{E_k(h); k\in\Z\}\cap{\mathcal D}_h(C_0)
$$
where the
 $E_k(h)$'s are complex numbers that satisfy,
\be
\label{reEk}
\re \,E_k(h)=\lambda_k(h)h^{2/3}-\frac{ \lambda_k(h)^2{\mathcal A}''(0)}{2{\mathcal A}'(0)}h^{4/3}+\ord (h^{5/3}),
\ee
\be
\label{imEk}
\im \,E_k(h)=-\frac{2\pi^2r_0(0)^2\left(\mu_1(\lambda_k(h))^2+\mu_2(\lambda_k(h))^2\right)}{{\mathcal A}'(0)}h^{5/3} +\ord(h^2),
\ee
uniformly as $h\to 0_+$.
Here, the function $\mu_1$ and $\mu_2$ are defined by,
$$
\begin{aligned}
\mu_1(t):=\int_0^{+\infty} \ai (\tau_1^{-2/3}(\tau_1y-t))\ai (-\tau_2^{-2/3}(\tau_2y+t))dy;\\
\mu_2(t):=\int_0^{+\infty} \ai (\tau_2^{-2/3}(\tau_2y-t))\ai (-\tau_1^{-2/3}(\tau_1y+t))dy,
\end{aligned}
$$
where $\ai$ stands for the usual Airy function.
\end{thm}
\begin{remark}\sl One can always choose $k=k(h)\to +\infty$ in such a way that $\lambda_k(h)\to \rho_0$ as $h\to 0_+$, where $\rho_0\in[-C_0,C_0]$ is fixed arbitrarily. In this case, (\ref{imEk}) gives
$$
\im \,E_k(h)= -\frac{2\pi^2r_0(0)^2(\mu_1(\rho_0)^2+\mu_2(\rho_0)^2)}{{\mathcal A}'(0)}h^{5/3}+o (h^{5/3}).
$$
In particular, if $(\mu_1(\rho_0), \mu_2(\rho_0))\not= (0,0)$, then the result produces an equivalent of the width of the resonance as $h\to 0_+$, and it is of order $h^{5/3}$. Let us observe that one has,
$$
\begin{aligned}
\mu_1(t)+\mu_2(t)& = \int_{-\infty}^{+\infty}\ai (\tau_1^{-2/3}(\tau_1y-t))\ai (-\tau_2^{-2/3}(\tau_2y+t))dy\\
& = (\tau_1+\tau_2)^{-1/3} \ai (-(\tau_1^{-1}+\tau_2^{-1})^{2/3}t),
\end{aligned}
$$
(where the last identity comes from an elementary computation involving the Fourier transform of $\ai$ and interpreting $\mu_1+\mu_1$ as a convolution of two Airy-type functions) and thus, the possible zeros of $\mu_1(\rho_0)^2+\mu_2(\rho_0)^2$ are among those of the function $t\mapsto \ai (-(\tau_1^{-1}+\tau_2^{-1})^{2/3}t)$. In particular, they are necessarily positive. In  the case where $\tau_1=\tau_2$, this phenomenon does occur exactly at the zeros of $\ai (-2^{\frac23}\tau_1^{-\frac23}t)$, and for these special values of $\rho_0$ the result just says that $\im \,E_k(h)$ is $o (h^{5/3})$.
\end{remark}
\begin{remark}\sl 
\label{rempseudo}
Under our assumptions, it can be proved that there exists an analytic distortion $\widetilde P_2$ of $P_2$ such that,  for $z\in {\mathcal D}_h(C_0)$, one has $\|(\widetilde P_2-z)^{-1}\|_{{\mathcal L}(L^2)} =\ord (h^{-1})$, while the corresponding analytic distortion $\widetilde P_1$ of $P_1$ satisfies $\|(\widetilde P_1-z)^{-1}\|_{{\mathcal L}(L^2)} =\ord ({\rm dis}(z,\sigma(P_1))^{-1}$.
Using the ellipticity of $\widetilde P_1$ and $\widetilde P_2$, and denoting by $\widetilde W$ and ${\widetilde W}^*$ the distorted operators obtained form $W$ and $W^*$, we deduce that, if  in addition $\min |z-\lambda_k(h)|\geq \delta h$ for some $\delta >0$ constant, and if $\sup(|r_0|+|r_1|)$ is sufficiently small, then,
$$
\|K(z)\|:=\|h^2(\widetilde P_2 -z)^{-1}\widetilde W^*(\widetilde P_1 -z)^{-1}\widetilde W\|\leq 1/2.
$$
Therefore, observing that the equation $(\widetilde P-z)u=v$ is equivalent to $u_1=(\widetilde P_1 -z)^{-1}(v_1-h\widetilde W u_2)$ and $(1-K(z))u_2=(\widetilde P_2 -z)^{-1}(v_2-h\widetilde W^*(\widetilde P_1 -z)^{-1}v_1)$, we conclude that if the distance between $z\in{\mathcal D}_h(C_0)$ and the eigenvalues of $\widetilde P$ (that are the resonances of $P$) is at least of order $h$, then one has,
\be
(\widetilde P -z)^{-1}=\ord (h^{-1}).
\ee
Since the resonances of $P$ are separated from each other by a distance of order $h$,  one can apply the  standard perturbation theory and conclude that if $\widetilde B$ is a uniformly bounded operator on $L^2\oplus L^2$, then the eigenvalues of $\widetilde P +h^2\widetilde B$ in ${\mathcal D}_h(C_0)$ differ from those of $\widetilde P$ by $\ord(h^2)$.
In particular, in this situation our result remains valid if we perturb $P$ by a semiclassical pseudodifferential operator $h^2B=h^2b(x,hD_x)$ where $b$ is  a bounded analytic symbol  in $\Gamma \times\Gamma$. This is typically what happens after the Feshbach reduction in the Born-Oppenheimer reduction process: see, e.g., \cite{MaMe, KMSW, MaSo}.
\end{remark}

\section{Fundamental solutions}
\label{sect3}
In order not to complicate too much the notations, we write the proof in the case
 $\tau_1=\tau_2=1$ only, but it is clear that all our treatment can be performed with any other positive values of these parameters. At the end of the proof, we explain where the changes occur exactly.

From now on we fix $\theta >0$ small enough, and in the sequel we will use the following notations :
$$
I_L:= (-\infty, 0]=\R_-\quad ; \quad I_R^\theta:=F_\theta ([0, +\infty))=F_\theta (\R_+),
$$
where $F_\theta (x):= x+i\theta f(x)$, with $f\in C^{\infty}(\R_+,\R_+)$, $f(x)=x$ for $x$ large enough, $f(x)=0$ for $x\in[0,x_\infty]$ for some $x_\infty>0$, and $f$ is chosen in such a way that, for any $x\geq x_\infty$, one has,
\be
\label{decItheta}
\im \int_{x_\infty}^{F_\theta (x)}\sqrt{ E-V_2(t)} dt \geq -Ch,
\ee
for some positive constant $C$, where the first integral is taken along $I_R^\theta$ (observe that, for $E\in {\mathcal D}_h(C_0)$ and $C>0$ sufficiently large, one always has $\im \int_{Ch^{\frac23}}^{x_\infty}\sqrt{ E-V_2(t)} dt=\ord (h)$).
Performing a Taylor expansion of this quantity as $\theta \to 0_+$, and using the fact that, for any $k\geq 1$, one has $t^kV_2^{(k)}(t) \to 0$  as $t\to +\infty$, we see that it is sufficient to choose $x_\infty$ sufficiently large and that $f$ satisfies,
$$
f'(t)\geq \delta (t^{-1}f(t)+t^{-3}f(t)^3)
$$ 
for some $\delta >0$ constant. By taking $f$ strictly increasing on $[x_\infty, +\infty)$, we see that the only possible problem is near $x_\infty$, but there one can take for instance $f(s):=e^{-1/(x-x_\infty)^2}$.

\subsection{Fundamental solutions on $I_L$}
\label{sect3}

On $I_L := (-\infty, 0]$, and for $E\in {\mathcal D}_h(C_0)$ and  $j=1,2$, let $u_{j,L}^\pm$ be the solutions to 
$(P_j - E)u = 0$ constructed in Appendix \ref{appendix2} (in particular, $u_{j,L}^-$ decays exponentially at $-\infty$, while $u_{j,L}^+$ grows exponentially).
Then, the Wronskian $\W[u_{j,L}^-, u_{j,L}^+]$ is independent of the variable $x$ and satisfies 
\begin{equation}\label{wronsL}
\W[u_{j,L}^-, u_{j,L}^+] = \frac{-2}{\pi h^{\frac23}}(1+\ord(h)) \qquad (h \to 0).
\end{equation}

For any $k\geq 0$ integer, we set,
$$
C^k_b(I_L) := \{ u\, :\ I_L \to \C \mbox{ of class } C^k \,;\, \sum_{0\leq j\leq k}\sup_{x\in I_L}|u^{(j)}(x)| < +\infty \},
$$
equipped with the norm $\pal u \pal_{C^k_b(I_L)} := \sum_{0\leq j\leq k}\sup_{I_L}|u^{(j)}|$, 
and we define a fundamental solution
$$K_{j,L}\, :\, C^0_b(I_L) \to C^2_b(I_L)\qquad (j=1,2),$$
of $P_j - E$ on $I_L$ by setting, for $v\in C^0_b(I_L)$,
\be
\label{eq1}
\begin{aligned}
K_{j,L}[v](x) := \frac{u_{j,L}^+(x)}{h^2 \W[u_{j,L}^+,u_{j,L}^-]} \int_{-\infty}^x & u_{j,L}^-(t)v(t)\,dt \\
& + \frac{u_{j,L}^-(x)}{h^2 \W[u_{j,L}^+,u_{j,L}^-]} \int_x^{0}\!\!\!\! u_{j,L}^+(t)v(t)\,dt.
\end{aligned}
\ee
Then, $K_{j,L}$ satisfies,
$$
(P_j-E)K_{j,L}={\mathbf 1},
$$
and, because of the form of the operator $W$, an integration by parts shows that we also have,
$$K_{j,L}W,\,\, K_{j,L}W^*\, :\, C^0_b(I_L) \to C^0_b(I_L)\qquad (j=1,2).$$

In view of the construction of solutions to the system, we prove,
\begin{proposition}\sl
\label{claim} As $h$ goes to $0_+$,
one has,
\be
\label{estnormK2L}
\pal hK_{2,L}W^* \pal_{{\mathcal L}(C^0_b(I_L))} = \ord (h^{\frac{1}{3}});
\ee
\be
\label{estnormM1L}
\pal h^2K_{1,L}WK_{2,L}W^* \pal_{{\mathcal L}(C^0_b(I_L))} = \ord (h^{\frac23}).
\ee
\end{proposition} 
\begin{proof} For $j=1,2$, we set,
\be
\begin{aligned}
\label{defUj}
& \widetilde u_{j,L}^\pm(x) := \max \{|u_{j,L}^\pm(x)|, |h\partial_xu_{j,L}^\pm(x)|\};\\
& U_j(x,t):= \widetilde u_{j,L}^+(x) \widetilde u_{j,L}^-(t){\bf 1}_{\{t<x\}}+  \widetilde u_{j,L}^-(x) \widetilde u_{j,L}^+(t){\bf 1}_{\{t>x\}} =U_j(t,x).
\end{aligned}
\ee
Thanks to our choice of $K_{j,L}$, and doing an integration by parts,  we see that we have,
\be
\label{estK2L}
\begin{aligned}
& |hK_{1,L}Wv(x)|=\ord (h^{-\frac13})\left( \int_{-\infty}^0U_1(x,t)|v(t)|dt + hU_1(x,0)|v(0)|\right);\\
& |hK_{2,L}W^*v(x)|=\ord (h^{-\frac13})\left( \int_{-\infty}^0U_2(x,t)|v(t)|dt + hU_2(x,0)|v(0)|\right).
\end{aligned}
\ee
In particular,
\be
\label{normK2L}
\pal hK_{2,L}W^* \pal=\ord (h^{-\frac13})\sup_{x\in I_L} \int_{-\infty}^0U_2(x,t)dt + \ord (h^{\frac23})\sup_{x\in I_L} U_2(x,0).
\ee
Moreover, using  the asymptotics of $u_{2,L}^\pm$ and $h\partial_xu_{2,L}^\pm$ on $I_L$, and fixing some constant $C_1>0$ sufficiently large, we obtain,
\begin{itemize}
\item If $x,t\leq -C_1h^{\frac23}$, then,
\be
\label{prop1U2}
U_2(x,t)=\ord (h^{\frac13})\frac{e^{-\left| \re \int_t^x  (V_2-E)^{1/2}\right| /h}}{|(V_2(x)-E)(V_2(t)-E)|^{\frac14}};
\ee
\item If $t\leq -C_1h^{\frac23}\leq x\leq 0$, then,
\be
\label{prop2U2}
U_2(x,t)=\ord (h^{\frac16})\frac{e^{-\left| \re \int_0^t  (V_2-E)^{1/2}\right| /h}}{|(V_2(t)-E)|^{\frac14}};
\ee
\item If $x,t\in [-C_1h^{\frac23}, 0]$,  then $U_2(x,t)=\ord (1)$.
\end{itemize}
Hence  $U_2(x,t)=\ord(1)$ uniformly, and when $x\leq -\delta$ with $\delta >0$ constant, there exists a positive constant $\alpha$ such that,
$$
\int_{-\infty}^0U_2(x,t)dt=\int_{-\infty}^{-\delta /2}e^{-\alpha |x-t|/h}dt +\ord (e^{-\alpha /h})=\ord (h).
$$
On the other hand, if $\delta$ is chosen sufficiently small and $x\in[-\delta, -C_1h^{\frac23}]$, then, there exists a (different) positive constant $\alpha$ such that,
$$
\begin{aligned}
\int_{-\infty}^0U_2(x,t)dt &=\int_{-2\delta}^{-C_1h^{\frac23}}U_2(x,t)dt +\ord (h^{\frac23})\\
&=\ord (h^{\frac13}|x|^{-\frac14})\int^{2\delta}_{C_1h^{\frac23}}\frac{e^{-\alpha \left|t^{\frac32}-|x|^{\frac32}\right|/h}}{t^{\frac14}}dt+\ord (h^{\frac23}).
\end{aligned}
$$
Setting $t=(hs)^{\frac23}$ in the integral, we obtain,
$$
\int_{-\infty}^0U_2(x,t)dt=\ord(h^{\frac13}|x|^{-\frac14}h^{\frac23-\frac16})\int_1^\infty \frac{e^{-\alpha \left|s-h^{-1}|x|^{\frac32}\right|}}{\sqrt s}ds+\ord(h^{\frac23})=\ord(h^{\frac23}).
$$
Finally, when $x\in [-C_1h^{\frac23},0]$, we have,
$$
\begin{aligned}
\int_{-\infty}^0U_2(x,t)dt &=\int_{-\delta}^{-C_1h^{\frac23}}U_2(x,t)dt +\ord (h^{\frac23})\\
&=\ord (h^{\frac16})\int^{\delta}_{C_1h^{\frac23}}\frac{e^{-\alpha t^{\frac32}/h}}{t^{\frac14}}dt+\ord (h^{\frac23}) =\ord(h^{\frac23}).
\end{aligned}
$$
Thus, we have prove,
\be
\label{intU2}
\sup_{x\leq 0}\int_{-\infty}^0 U_2(x,t)dt =\ord (h^{\frac23}),
\ee
and, by (\ref{normK2L}) (and the fact that $U_2=\ord (1)$), (\ref{estnormK2L}) follows.

Now, let us prove the estimate for $M_L:= h^2K_{1,L}WK_{2,L}W^*$. We see on the definition of $K_{1,L}$ and on (\ref{estK2L}) that we have,
\be
\label{decompML}
\begin{aligned}
|M_Lv(x)|\leq & Ch^{-\frac23}\int_{-\infty}^0\int_{-\infty}^0 U_1(x,t)U_2(t,s)|v(s)|ds dt \\
& +Ch^{\frac13}\int_{-\infty}^0 U_1(x,t)U_2(t,0)|v(0)|dt\\
& +Ch^{\frac13}U_1(x,0)\int_{-\infty}^0 U_2(0,t)|v(t)|dt\\
& +Ch^{\frac43}U_1(x,0)U_2(0,0)|v(0)|.
\end{aligned}
\ee
Using (\ref{intU2}) and the fact that  $U_j=\ord (1)$ uniformly ($j=1,2$), we see that the last three terms are $\ord (h) \sup_{I_L}|v|$ (observe that $U_j(t,x)=U_j(x,t)$ for all $x,t$).

In order to estimate the integral, we use the following properties of $U_1$:
For any $\delta>0$ small enough, there exists $\alpha >0$ constant, such that,
\begin{itemize}
\item If $x,t\leq x^*-\delta$, then,
$$
U_1(x,t)=\ord (h^{\frac13})e^{-\alpha |t-x|/h};
$$
\item If $t\leq x^*-2\delta$ and $x\in[x^*-\delta, 0]$, then,
$$
U_1(x,t)=\ord (h^{\frac16}e^{-\alpha /h});
$$
\item If $x\in [x^*-4\delta,0]$ and $t\in [-\delta, -C_1h^{\frac23}]$,  then $U_1(x,t)=\ord (h^{\frac16}|t|^{-\frac14})$;
\item If $x\in [x^*-4\delta,0]$ and $t\in [-C_1h^{\frac23},0]\cup [x^*-4\delta, -\delta]$,  then $U_1(x,t)=\ord (1)$.
\end{itemize}

Moreover, by the properties of $U_2$, we also know that any part of the integral corresponding to $|t-s|\geq \delta$ with $\delta >0$ constant  is exponentially small.

We first consider the case $x\in (-\infty, x^*-2\delta]$ for some small positive constant $\delta$. 

Then,  we see that there exists a constant $\alpha >0$ such that,
$$
\begin{aligned}
\iint_{-\infty}^0 U_1(x,t)U_2(t,s)dtds &=\ord (h^{\frac23})\int_{-\infty}^{x^*-\delta}dt \int_{-\infty}^{x^*-\delta/2} e^{-\alpha (|t-x|+|s-t|)/h}ds\\
&\hskip 6cm+\ord(e^{-\alpha /h})\\
& =\ord(h^{2+\frac23}).
\end{aligned}
$$

Now, when $x\in[x^*-2\delta,0]$, and still denoting by $\alpha$ every new positive constant that may appear, we have,
$$
\iint_{-\infty}^0 U_1(x,t)U_2(t,s)dtds =\int_{x^*-3\delta}^0dt \int_{x^*-4\delta}^0U_1(x,t)U_2(t,s)ds+\ord(e^{-\alpha /h}),
$$
and,
$$
\begin{aligned}
\int_{x^*-3\delta}^0dt \int_{x^*-4\delta}^0& U_1(x,t)U_2(t,s)ds \\
=&\ord(h^{\frac13})\int_{x^*-3\delta}^{-\delta}dt\int_{x^*-4\delta}^{-\delta/2}e^{-\alpha|t-s|/h}ds +\ord(e^{-\alpha/h})\\
&+\ord(h^{\frac12})\int_{-\delta}^{-C_1h^{\frac23}}dt\int_{-2\delta}^{-C_1h^{\frac23}}\frac{e^{-\alpha \left| |t|^{\frac32}-|s|^{\frac32}\right|/h}}{|t|^{\frac12}|s|^{\frac14}}ds\\
&+\ord(h^{\frac16})\int_{-\delta}^{-C_1h^{\frac23}}dt\int_{-C_1h^{\frac23}}^0\frac{e^{-\alpha |t|^{\frac32}/h}}{|t|^{\frac14}}ds\\
&+\ord(h^{\frac16})\int_{-C_1h^{\frac23}}^0 dt\int_{-\delta}^{-C_1h^{\frac23}}\frac{e^{-\alpha |s|^{\frac32}/h}}{|s|^{\frac14}}ds
+\ord(h^{\frac43}).
\end{aligned}
$$
Hence,
\be
\label{estinterm1}
\begin{aligned}
\int_{x^*-3\delta}^0dt \int_{x^*-4\delta}^0U_1(x,t)U_2(t,s)ds =&\ord(h^{\frac12})\int^{\delta}_{C_1 h^{\frac23}}dt\int^{2\delta}_{C_1h^{\frac23}}\frac{e^{-\alpha \left| t^{\frac32}-s^{\frac32}\right|/h}}{t^{\frac12}s^{\frac14}}ds\\
&+\ord(h^{\frac56})\int^{\delta}_{C_1h^{\frac23}}\frac{e^{-\alpha t^{\frac32}/h}}{t^{\frac14}}dt
+\ord(h^{\frac43}).
\end{aligned}
\ee
For the first term, the change of variables $(t,s)\mapsto (t^{2/3},s^{2/3})$ gives,
$$
\int^{\delta}_{C_1h^{\frac23}}dt\int^{2\delta}_{C_1h^{\frac23}}\frac{e^{-\alpha \left| t^{\frac32}-s^{\frac32}\right|/h}}{t^{\frac12}s^{\frac14}}ds
={\mathcal O}(1)\iint_{C_2 h}^{\delta'}\frac{e^{-\alpha |t-s|/h}}{t^{2/3} s^{1/2}}dtds,
$$
with $C_2:=C_1^{2/3}$ and $\delta':=(2\delta)^{2/3}$.
Dividing this integral in two parts, depending whether $t\leq s$ or $s\leq t$, and first integrating with respect to the larger of the two variables, we obtain,
\be
\begin{aligned}
\label{estintspec}
\int^{\delta}_{C_1h^{\frac23}}dt\int^{2\delta}_{C_1h^{\frac23}}\frac{e^{-\alpha \left| t^{\frac32}-s^{\frac32}\right|/h}}{t^{\frac12}s^{\frac14}}ds
 &={\mathcal O}(1)\int_{C_2 h}^{\delta'}dt \frac{e^{\alpha t/h}}{t^{7/6}}\int_t^{\delta'}e^{-\alpha s/h}ds\\
 & ={\mathcal O}(h^{5/6}). 
 \end{aligned}
\ee
Moreover, a simple change of variable gives,
$$
\int^{\delta}_{C_1h^{\frac23}}\frac{e^{-\alpha t^{\frac32}/h}}{t^{\frac14}}dt=\ord (h^{\frac12}).
$$
Inserting into (\ref{estinterm1}), we deduce that, for $x\in[x^*-2\delta,0]$, we have,
\be
\label{estdbleintU2}
\iint_{-\infty}^0 U_1(x,t)U_2(t,s)dtds =\ord (h^{4/3}).
\ee
Finally, going back to (\ref{decompML}), we conclude (\ref{estnormM1L}).
\end{proof}

\subsection{Fundamental solutions on $I_R^\theta$}

On $\R_+$,  let $u_{1,R}^\pm$ (resp. $u_{2,L}^\pm$) be the solutions to
$(P_1 - E)u = 0$ (resp. $(P_2-E)u=0$) constructed in Appendix \ref{appendix2}. Setting,
\be
\label{defu2Rpm}
u_{2,R}^\pm := 2^{\pm \frac12}e^{i\frac{\pi}4}(\frac12 a_2^-u_{2,L}^-\pm ia_2^+ u_{2,L}^+)
\ee
by Proposition \ref{yaf2} we have,
$$
u_{2,R}^\pm(x) \sim (1+\ord (h)) \frac{h^{\frac16}e^{i\frac{\pi}4}}{\sqrt{2\pi}}(E-V_2(x))^{-1/4} e^{\mp i \int_0^x\sqrt{ E-V_2(t)}\, dt/h}\qquad (x \to + \infty).
$$

Then, all these solutions can be extended as holomorphic functions in a complex sector around $\R_+$, and in particular (thanks to (\ref{decItheta})), $u_{j,R}^-$  ($j=1,2$) remains uniformly bounded on $I_R^\theta$, and exponentially decaying at infinity along $I_R^\theta$. Moreover,  the Wronskians $\W[u_{j,R}^-, u_{j,R}^+]$ are independent of the variable $x$ and satisfy, 
\begin{equation}\label{wronsR}
\W[u_{j,R}^-, u_{j,R}^+] = \frac{2}{\pi h^{\frac23}}(1+\ord(h)) \qquad (h \to 0).
\end{equation}
We define a fundamental solution of $P_j - E$ on $I_R^\theta$ as,
\be\label{eq1R}
\begin{aligned}
K_{j,R}[v](x) :=  \frac{u_{j,R}^-(x)}{h^2 \W[u_{j,R}^-,u_{j,R}^+]} &\int_{0}^x u_{j,R}^+(t)v(t)\,dt \\
&+ \frac{u_{j,R}^+(x)}{h^2 \W[u_{j,R}^-,u_{j,R}^+]} \int_x^{+\infty}\!\!\!\! u_{j,R}^-(t)v(t)\,dt,
\end{aligned}
\ee
where  $v$ is in the space $C_b^0(I_R^\theta)$ of bounded continuous functions on $I_R^\theta$, and the integrals run over $I_R^\theta$.

Then, using the semiclassical asymptotic behaviour of $u_{j,R}^\pm$ on $I_R^\theta$, one can prove exactly as for (\ref{estnormK2L})-(\ref{estnormM1L}) ($x_\infty$ playing the role of $x^*$) that we have,
\begin{proposition}\label{claim2}\sl One has,
\be
\label{estnormK1R}
\pal hK_{1,R}W \pal_{{\mathcal L}(C_b^0(I_R^\theta))} = \ord ( h^{\frac{1}{3}});
\ee
\be
\label{estnormM2R}
\pal h^2K_{2,R}W^*K_{1,R}W \pal_{{\mathcal L}(C_b^0(I_R^\theta))} = \ord (h^{\frac23}),
\ee
uniformly as $h$ goes to $0_+$.
\end{proposition}

\section{Solutions on $I_L$ and $I_R^\theta$}
\label{sect4}

In this section, we use the previous constructions of fundamental solutions in order to construct solutions to the system (\ref{sch}) in $I_L$ and $I_R^\theta$.

We first consider the interval $I_L$. We set,
$$
w_{1,L}^0:=\left(\begin{array}{c} u_{1,L}^-\\0\end{array}\right)\quad ;\quad w_{2,L}^0:=\left(\begin{array}{c} 0\\u_{2,L}^-\end{array}\right),
$$
and we look for solutions $u:=\left(\begin{array}{c} u_1\\u_2\end{array}\right)$ to (\ref{sch}) in $I_L$, that are close to $w_{1,L}^0$ or $w_{2,L}^0$ as $h\to 0$. We rewrite (\ref{sch}) as,
$$
\left\{
\begin{array}{l}
(P_1-E)u_1 =-hWu_2;\\
(P_2-E)u_2=-hW^*u_1.
\end{array}
\right.
$$
Assuming that $u_1$ is in $C^0_b (I_L)$, we set $u_2=-hK_{2,L} W^*u_1$, so that the system reduces to,
$$
(P_1-E)u_1=h^2WK_{2,L}W^*u_1,
$$
and a solution will be given by any $u_1\in C^0_b (I_L)$ such that,
$$
u_1 =u_{1,L}^-+h^2K_{1,L}WK_{2,L}W^*u_1.
$$
Now, by Proposition \ref{claim}, the operator $M_L:=h^2K_{1,L}WK_{2,L}W^*$ is $\cO(h^{\frac23})$ when acting on $C^0_b (I_L)$. By construction, we also know that $u_{1,L}^-$ is in $C^0_b (I_L)$. Therefore, we can define,
\be
\label{w1L}
w_{1,L}:=\left(\begin{array}{c} \sum_{j\geq 0}M_L^ju_{1,L}^-\\-hK_{2,L} W^*\sum_{j\geq 0}M_L^ju_{1,L}^-\end{array}\right)\in C^0_b (I_L)^2,
\ee
and $w_{1,L}$ is solution to (\ref{sch}) in $I_L$, with $w_{1,L}\to w_{1,L}^0$ as $h\to 0_+$. In a similar way,  we can define,
\be
\label{w2L}
w_{2,L}:=\left(\begin{array}{c} -\sum_{j\geq 0}M_L^j (hK_{1,L} Wu_{2,L}^-)\\ u_{2,L}^-+hK_{2,L} W^*\sum_{j\geq 0}M_L^j(hK_{1,L} Wu_{2,L}^-)\end{array}\right) \in C^0_b (I_L)^2,
\ee
that is solution to (\ref{sch}) in $I_L$, with $w_{2,L}\to w_{2,L}^0$ as $h\to 0_+$. 

\begin{rem}\sl
The standard WKB method (see, e.g., \cite{FuRa, Vo}) gives us asymptotic expansions for $u_{j,L}^-$ ($j=1,2$) inside $(x^*,0)$, of the form,
\be
\begin{split}
\label{asympujL}
u_{1,L}^-\sim &\frac{2h^{\frac16}}{\sqrt\pi}(E-V_1(x))^{-1/4}\sin\left( \frac{{\mathcal A}(E)+\nu_1(x)}h+\frac{\pi}4\right)\left(1+\sum_{k\geq 1}a_{1,k}(x)h^k\right) \\
& +\frac{2h^{\frac16}}{\sqrt\pi}(E-V_1(x))^{-1/4}\cos\left( \frac{{\mathcal A}(E)+\nu_1(x)}h+\frac{\pi}4\right)\sum_{k\geq 1}b_{1,k}(x)h^k; \\
u_{2,L}^-\sim &\frac{h^{\frac16}}{\sqrt\pi}(V_2(x)-E)^{-1/4}
e^{-\frac1{h}\int_x^{x_2(E)}\sqrt{V_2(t)-E}\, dt}(1+\sum_{k\geq 1}a_{2,k}(x)h^k),
\end{split}
\ee
where $x_2(E)$ is the unique solution of $V_2(x)=E$ close to 0,  $\nu_1(x):= \int_{x_1(E)}^x \sqrt{E-V_1(t)}\, dt$.
\end{rem}

Now, on $I_R^\theta$, a similar construction can be done by using Proposition \ref{claim2}, and by starting from the solutions  $u_{j,R}^-$ ($j=1,2$) to $(P_j-E)u_{j,R}^- =0$ on $I_R^\theta$.
 Setting $M_R:=h^2K_{2,R}W^*K_{1,R}W$, by Proposition \ref{claim2} we have $\| M_R\| ={\mathcal O}(h^{\frac23})$, and thus, defining,
\be
\label{w1R}
w_{1,R}:=\left(\begin{array}{c} u_{1,R}^-+hK_{1,R}W\sum_{j\geq 0}M_R^j(hK_{2,R}W^*u_{1,R}^-)\\
-\sum_{j\geq 0}M_R^j(hK_{2,R} W^*u_{1,R}^-)
\end{array}\right)\in C_b^0(I_R^\theta)^2;
\ee
\be
\label{w2R}
w_{2,R}:=\left(\begin{array}{c} 
-hK_{1,R} W\sum_{j\geq 0}M_R^ju_{2,R}^-\\ 
\sum_{j\geq 0}M_R^ju_{2,R}^-
\end{array}\right)\in C_b^0(I_R^\theta)^2,
\ee
we see that they are both solutions to  (\ref{sch}) on $I_R^\theta$, and they respectively  approach $w_{1,R}^0:=\left(\begin{array}{c} u_{1,R}^-\\0\end{array}\right)$ and $w_{2,R}^0:=\left(\begin{array}{c} 0\\ u_{2,R}^-\end{array}\right)$, as $h\to 0_+$. 

\begin{rem}\sl
Still by the standard WKB method, we have asymptotic expansions for $u_{j,R}^-$ ($j=1,2$) inside $(0,+\infty)$, of the form,
\be
\begin{split}
&u_{1,R}^-\sim \frac{\pi^{\frac12}h^{\frac16}}{2}(V_1(x)-E)^{-1/4}
e^{-\frac1{h}\int_{x_1(E)}^x\sqrt{V_1(t)-E}\, dt}(1+\sum_{k\geq 1}b_{1,k}(x)h^k);
\\
&u_{2,R}^-\sim \frac{\pi^{\frac12}h^{\frac16}}{2}(E-V_2(x))^{-1/4}
e^{\frac{i}{h}\int_{x_2(E)}^x\sqrt{E-V_2(t)}\, dt}(1+\sum_{k\geq 1}b_{2,k}(x)h^k).
\end{split}
\ee
\end{rem}
The solutions we have just constructed are not only bounded, and actually we have,
\begin{proposition}\sl
The solutions $w_{j,L}$ given by (\ref{w1L})-(\ref{w2L}), and $w_{j,R}$ given by (\ref{w1R})-(\ref{w2R})  ($j=1,2$) satisfy,
$$
w_{j,L} \in L^2(I_L)\oplus L^2(I_L)\quad ; \quad w_{j,R} \in L^2(I_R^\theta)\oplus L^2(I_R^\theta).
$$
\end{proposition}
\begin{proof}
It is sufficient to prove that, for any $N\geq 1$, one has $w_{j,S}=\ord (\langle x\rangle^{-N})$ as $|x|\to \infty$ (on $I_L$ or $I_R^\theta$, depending if $S=L$ or $S=R$). But thanks to the exponential decay of $U_j(x,t)$ ($j=1,2$) as $|x-t|\to\infty$, $|x|\gg 1$, we immediately see that (\ref{intU2}) and (\ref{estdbleintU2}) remain valid with $U_j(x,t)$ replaced by $\langle x\rangle^N U_j(x,t)\langle t\rangle^{-N}$. As a consequence, the estimates of Proposition \ref{claim}  extend to the operators $\langle x\rangle^NhK_{2,L}W^*\langle x\rangle^{-N}$ and $\langle x\rangle^Nh^2K_{1,L}WK_{2,L}W^*\langle x\rangle^{-N}$, and since $\langle x\rangle^Nu_{j,L}^-\in C_b^0(I_L)$ ($j=1,2$), the result for $w_{j,L}$ follows. The same arguments apply in $I_R^\theta$, and the result for $w_{j,R}$ follows, too.
\end{proof}

Now, by general theory on systems of ordinary differential equations, we know that the space of solutions to (\ref{sch}) that are $L^2$ in $I_L$ (resp. in $I_R^\theta)$ is at most two-dimensional. As a consequence,  
the previous proposition implies,
\begin{proposition}\label{quantcondi1}
$E$ is a resonance if and only if the four solutions
$w_{1,L}$, $w_{2,L}$, 
$w_{1,R}$ and $w_{2,R}$ are linearly dependent.
\end{proposition}

\section{Estimates at the crossing point}
\label{sect5}
In this section, we investigate the asymptotic behaviours of $w_{j,L}(x)$, $w_{j,R}(x)$, and their first derivative at $x=0$.

 We first prove,
\begin{proposition}\sl 
\label{u1Losc}
Let $x_0\in (x^*,0)$. Then, for $x\in [x_0,0]$, one has,
\be
\label{LfunctR}
u_{1,L}^\pm= a_\pm u_{1,R}^- + b_\pm u_{1,R}^+,
\ee
with,
$$
\begin{aligned}
& a_-=\sin\frac{{\mathcal A}(E)}h+\ord(h)\qquad ;\quad  b_- =2\cos\frac{{\mathcal A}(E)}h +\ord (h);\\
& a_+=\frac12 \cos\frac{{\mathcal A}(E)}h+\ord(h)\quad ;\quad  b_+=-\sin\frac{{\mathcal A}(E)}h +\ord (h),
 \end{aligned}
 $$
uniformly as $h\to 0_+$.
\end{proposition}
\begin{proof} Since $(u_{1,R}^-,u_{1,R}^+)$ is a basis of solutions to $(P_1-E)u=0$, we know that (\ref{LfunctR}) is verified with,
$$
a_-=\frac{\W (u_{1,L}^-,u_{1,R}^+)}{\W(u_{1,R}^-,u_{1,R}^+)}\quad ;\quad b_-=\frac{\W (u_{1,R}^-, u_{1,L}^-)}{\W(u_{1,R}^-,u_{1,R}^+)}.
$$
We compute these Wronskians at some arbitrary point in $x\in (x^*,0)$, by using formula (\ref{asympujL}), Proposition \ref{yaf1} and (\ref{wronsR}). By (\ref{asympujL}), we have,
$$
\begin{aligned}
& u_{1,L}^-(x) =\frac{2h^{\frac16}}{\sqrt{\pi}(E-V_1(x))^{\frac14}}\cos \left( \frac{{\mathcal A}(E)+\nu_1(x)}{h} -\frac{\pi}4\right) + \ord (h^{\frac76});\\
& (u_{1,L}^-)'(x) =-\frac{2h^{\frac16}(E-V_1(x))^{\frac14}}{h\sqrt{\pi}}\sin \left( \frac{{\mathcal A}(E)+\nu_1(x)}{h} -\frac{\pi}4\right) + \ord (h^{\frac16}).
\end{aligned}
$$
Therefore, using Proposition \ref{yaf1} and (\ref{uRprime}), we obtain,
$$
\begin{aligned}
\W(u_{1,L}^-, u_{1,R}^+) = & \frac{2h^{-\frac12}(\xi_1')^{\frac12}}{\sqrt{\pi}(E-V_1)^{\frac14}}\cos \left( \frac{{\mathcal A}(E)+\nu_1(x)}{h} -\frac{\pi}4\right)\bi'(h^{-\frac23}\xi_1) \\
& +\frac{2h^{-\frac56}(E-V_1)^{\frac14}}{\sqrt{\pi}(\xi_1')^{\frac12}}\sin \left( \frac{{\mathcal A}(E)+\nu_1(x)}{h} -\frac{\pi}4\right)\bi(h^{-\frac23}\xi_1)\\
& +\ord(h^{\frac13}).
\end{aligned}
$$
Then, using the asymptotic behaviour of $\bi$ and $\bi'$ at $-\infty$, and observing that one has the identity $\frac23 (-\xi_1)^{\frac32} = -\nu_1$ (so that $(\xi_1')^{\frac12} =(E-V_1)^{1/4}(-\xi_1)^{-1/4}$), this gives,
$$
\begin{aligned}
W(u_{1,L}^-, u_{1,R}^+) = &\frac{2h^{-\frac23}}{\pi}\cos \left( \frac{{\mathcal A}(E)+\nu_1(x)}{h} -\frac{\pi}4\right) \cos\left(\frac{\nu_1}h +\frac{\pi}4\right) \\
&+ \frac{2h^{-\frac23}}{\pi}\sin \left( \frac{{\mathcal A}(E)+\nu_1(x)}{h} -\frac{\pi}4\right) \sin\left(\frac{\nu_1}h +\frac{\pi}4\right)  +\ord(h^{\frac13})\\
& = \frac{2h^{-\frac23}}{\pi}\sin\frac{{\mathcal A}(E)}h +\ord(h^{\frac13}),
\end{aligned}
$$
and thus, by (\ref{wronsR}),
$$
a_-=\sin\frac{{\mathcal A}(E)}h+\ord(h).
$$
The same arguments hold for $b_-, a_+,b_+$, and the result follows.
\end{proof}

From now on, we set,
\be
\widetilde\partial:=\eps^2\partial_x,
\ee
and we will use this operator in all the Wronskians that will appear (instead of the usual derivative), denoting them by $\widetilde{\mathcal W}$ instead of ${\mathcal W}$.

\begin{proposition}\sl
\label{calculen0}
For $j=1,2$ and $S=L,R$, there exist complex numbers $\alpha_{j,S}, \beta_{j,S}$, such that,
\be
\label{estw1S}
\begin{aligned}
& w_{1,S} (0) =  \left[\begin{array}{c}
u_{1,S}^-(0) + \beta_{1,S} u_{1,S}^+(0)\\
\alpha_{1,S} u_{2,S}^+(0)
\end{array}\right] +\ord(h);\\
& \widetilde\partial w_{1,S}'(0) =\left[\begin{array}{c}
\widetilde\partial u_{1,S}^-(0) + \beta_{1,S} \widetilde\partial u_{1,S}^+(0)\\
\alpha_{1,S} \widetilde\partial u_{2,S}^+(0)
\end{array}\right]  +\ord(h),
\end{aligned}
\ee
\be
\label{estw2S}
\begin{aligned}
& w_{2,S} (0) =  \left[\begin{array}{c}
\alpha_{2,S} u_{1,S}^+(0)\\
u_{2,S}^-(0) + \beta_{2,S} u_{2,S}^+(0)
\end{array}\right] +\ord(h);\\
& \widetilde\partial w_{2,S}'(0) =\left[\begin{array}{c}
\alpha_{2,S} \widetilde\partial u_{1,S}^+(0)\\
\widetilde\partial u_{2,S}^-(0) + \beta_{2,S} \widetilde\partial u_{2,S}^+(0)
\end{array}\right]  +\ord(h),
\end{aligned}
\ee
uniformly as $h\to 0_+$. 
\end{proposition}
\begin{proof} We prove (\ref{estw1S}) and (\ref{estw2S}) in the case $S=L$ only (the case $S=R$ being similar). We start with $j=1$. By 
(\ref{w1L}), we have,
\be
w_{1,L}:=\left(\begin{array}{c} u_{1,L}^-+hK_{1,L}Wr_{1,L}\\-r_{1,L}\end{array}\right)+\ord (h),
\ee
with,
$$
r_{1,L}:=hK_{2,L} W^*u_{1,L}^-.
$$
In particular, using the expression (\ref{eq1}) of $K_{2,L}$, we find
$$
r_{1,L}(0)=-\alpha_{1,L}u_{2,L}^+(0),
$$
with,
\be
\label{defalpha1L}
\alpha_{1,L}:=\frac{-1}{h\W(u_{2,L}^+, u_{2,L}^-)}\int_{-\infty}^0u_{2,L}^-(t)(W^*u_{1,L}^-)(t)dt.
\ee
In the same way,
$$
hK_{1,L}Wr_{1,L} (0) =\beta_{1,L}u_{1,L}^+(0),
$$
with,
\be
\label{defbeta1L}
\beta_{1,L}:=\frac1{h\W(u_{1,L}^+, u_{1,L}^-)}\int_{-\infty}^0u_{1,L}^-(t)(Wr_{1,L})(t)dt.
\ee
In addition, we can write,
$$
r_{1,L}(x)= -\alpha_{1,L}u_{2,L}^+(x)+\widetilde r_{1,L}(x)
$$
with,
$$
\begin{aligned}
\widetilde r_{1,L}(x):=\frac1{h\W(u_{2,L}^+, u_{2,L}^-)}& \int_0^x\left( u_{2,L}^+(x) u_{2,L}^-(t) -u_{2,L}^-(x) u_{2,L}^+(t)\right)\\
& \hskip 4cm \times (W^*u_{1,L}^-)(t)dt.
\end{aligned}
$$
In particular $\widetilde r_{1,L}'(0)=0$, and thus $r_{1,L}'(0)= -\alpha_{1,L}(u_{2,L}^+)'(0)$. We see in the same way that $(hK_{1,L}Wr_{1,L})' (0) =\beta_{1,L}(u_{1,L}^+)'(0)$.

Moreover, concerning the remainder term in the derivative, we observe that for any $j,S$, one has $\widetilde\partial u_{j,S}^\pm (0) =\ord (1)$ (see, e.g., Remark \ref{derivees}). Therefore, for any $v\in C^0_b(I_L)$, we see on the definitions of $M_L$ and $K_{2,L}$ that we have,
$$
\begin{aligned}
& \widetilde\partial (hK_{2,L}W^*v)(0)= \ord(\eps^{-1})\left(\int_{-\infty}^0 |u_{2,L}^-(t)|.|v(t)| dt + h|(u_{2,L}^-)'(0)|.|v(0)|\right);\\
& \widetilde\partial M_Lv(0)=\ord (\eps^{-2})\left(\iint_{-\infty}^0|U_2(t,s)| .|v(s)|dtds + h\int_{-\infty}^0|U_2(t,0)|.|v(0)|dt\right),
\end{aligned}
$$
where $U_2$ is defined in (\ref{defUj}). Then, we observe that $\int_{-\infty}^0 |u_{2,L}^-(t)|dt =\ord(\eps^2)$ and, with the same proof as for (\ref{estdbleintU2}), we have $\iint_{-\infty}^0|U_2(t,s)|dtds=\ord (\eps^4)$. Using also (\ref{intU2}), we obtain,
$$
\begin{aligned}
& \widetilde\partial (hK_{2,L}W^*v)(0)= \ord(\eps)\sup_{I_L} |v|;\\
& \widetilde\partial M_Lv(0)=\ord (\eps^{2})\sup_{I_L} |v|.
\end{aligned}
$$
As a consequence,
$$
\begin{aligned}
& \widetilde\partial (hK_{2,L}W^*\sum_{j\geq 1}M_L^ju_{1,L}^-)(0)= \ord(\eps^3);\\
& \widetilde\partial (\sum_{j\geq 2}M_L^ju_{1,L}^-)(0)=\ord (\eps^{4}),
\end{aligned}
$$
and (\ref{estw1S}) follows. The proof of (\ref{estw2S}) is almost the same, with the only difference that the starting function is $hK_{1,L}Wu_{2,L}^-$ instead of $u_{1,L}^-$. But thanks to the decay properties of $u_{2,L}^-$, we see that its behaviour is similar to that of $u_{1,L}^-$, and (\ref{estw2S}) follows.

In addition to (\ref{defalpha1L})-(\ref{defbeta1L}), the other constants appearing in (\ref{estw1S})-(\ref{estw2S}) are,
\be
\label{defconstjS}
\begin{aligned}
& \alpha_{2,L}=  \frac{-\int_{-\infty}^0u_{1,L}^-(t)(Wu_{2,L}^-)(t)dt}{h\W(u_{1,L}^+, u_{1,L}^-)}\quad ; \quad 
\beta_{2,L}= \frac{\int_{-\infty}^0u_{2,L}^-(t)(W^*r_{2,L})(t)dt}{h\W(u_{2,L}^+, u_{2,L}^-)};\\
& \alpha_{1,R}=  \frac{-\int^{+\infty}_0u_{2,R}^-(t)(W^*u_{1,R}^-)(t)dt}{h\W(u_{2,R}^-, u_{2,R}^+)}\,\, ; \quad 
\beta_{1,R}=\frac{\int^{+\infty}_0u_{1,R}^-(t)(Wr_{1,R})(t)dt}{h\W(u_{1,R}^-, u_{1,R}^+)};\\
& \alpha_{2,R}=  \frac{-\int^{+\infty}_0u_{1,R}^-(t)(Wu_{2,R}^-)(t)dt}{h\W(u_{1,R}^-, u_{1,R}^+)} \quad ; \quad 
\beta_{2,R}=\frac{\int^{+\infty}_0u_{2,R}^-(t)(W^*r_{2,R})(t)dt}{h\W(u_{2,R}^-, u_{2,R}^+)},
\end{aligned}
\ee
where we have set,
$$
r_{2,L}:=hK_{1,L}Wu_{2,L}^- \quad ; \quad r_{1,R}:= hK_{2,R}W^*u_{1,R}^- \quad ; \quad  r_{2,R}:=hK_{1,R}Wu_{2,R}^-,
$$
and where, in the case of $\alpha_{j,R}$ and $\beta_{j,R}$, the integrals run over $I_R^\theta$.
\end{proof}

Now, setting,
\be
\begin{aligned}
& \mu_A(t):=\int_0^{+\infty} \ai (y-t)\check\ai (y+t)dy;\\
& \mu_B(t):=\int_0^{+\infty} \ai (y-t)\check\bi (y+t)dy,
\end{aligned}
\ee
we also have,
\begin{proposition}\sl
\label{calculen0suite}
As $h\to 0_+$, one has,
\be
\label{approxalphabeta}
\begin{aligned}
& \alpha_{j,L} =-2\eps\pi r_0(0)\left(\mu_A(\re\,\rho)\sin\frac{{\mathcal A(E)}}{h}+ \mu_B(\re\,\rho)\cos\frac{{\mathcal A(E)}}{h}\right)+\ord(\eps^2),\\
& \alpha_{j,R} =-\frac{\eps\pi r_0(0)e^{i\frac{\pi}4}}{\sqrt 2}\left( \mu_A(\re\,\rho)-i\mu_B(\re\,\rho)\right) +\ord (\eps^2),\quad (j=1,2);\\
& \im\, \beta_{1,L}=\ord (h);\\
&\im\, \beta_{1,R}=\pi^2r_0(0)^2\eps^2\left( \mu_A(\re\,\rho)^2+ \mu_B(\re\,\rho)^2\right) +\ord(h)\\
&\re\, \beta_{1,R} =2\eps^2\pi^2r_0(0)^2\iint_{0\leq s\leq t}\ai(t-\re\,\rho) \ai(s-\re\,\rho) \\
& \hskip 0.5cm\times \left(\check\ai(t+\re\,\rho)\check\bi(s+\re\,\rho)-\check\ai(s+\re\,\rho)\check\bi(t+\re\,\rho)\right)dsdt +\ord (h);\\
& \beta_{2,S} =\ord (\eps^2),\quad (S=L,R).
\end{aligned}
\ee
\end{proposition}
\begin{proof}

Let us  first study $\alpha_{1,L}$ given in (\ref{defalpha1L}). Using (\ref{wronsL}) and the exponential decay of $u_{2,L}^-$ away from 0, we obtain,
$$
\alpha_{1,L}=\frac{-\pi}{2\eps}\int_{-\delta}^0u_{2,L}^-(t)(r_0(t)u_{1,L}^-(t)+hD_t(\overline{r_1}u_{1,L}^-)(t))dt +\ord(\eps^2),
$$
where $\delta >0$ is arbitrarily small. Further, an integration by parts gives,
$$
\alpha_{1,L}=\frac{-\pi}{2\eps}\int_{-\delta}^0u_{2,L}^-(t)r_0(t)u_{1,L}^-(t)+ih(u_{2,L}^-)'(t)\overline{r_1}(t)u_{1,L}^-(t))dt +\ord(\eps^2),
$$
and since $h(u_{2,L}^-)'(t)u_{1,L}^-(t)$ is $\ord(h^{1/3}e^{-c|t|^{3/2}/h})$ on $[-\delta, -Ch^{2/3}]$ (with $c>0$, and $C>1$ large enough), and is $\ord(h^{1/3})$ on $[-Ch^{2/3},0]$, we are reduced to,
\be
\label{alpha1Lred}
\alpha_{1,L}=\frac{-\pi}{2\eps}\int_{-\delta}^0u_{2,L}^-(t)r_0(t)u_{1,L}^-(t)dt +\ord(\eps^2).
\ee
We introduce a large extra-parameter $\lambda \gg1$, and, dividing the integral in two parts, we set,
$$
\begin{aligned}
& \alpha_{1,L}^-(\lambda)=\frac{-\pi}{2\eps}\int_{-\delta}^{-\lambda h^{\frac23}}u_{2,L}^-(t)r_0(t)u_{1,L}^-(t)dt;\\
& \alpha_{1,L}^+(\lambda)=\frac{-\pi}{2\eps}\int_{-\lambda h^{\frac23}}^0u_{2,L}^-(t)r_0(t)u_{1,L}^-(t)dt.
\end{aligned}
$$
We have,
$$
\alpha_{1,L}^-(\lambda)=\ord(\eps^{-1})\int_{\lambda h^{\frac23}}^\delta \frac{h^{\frac13}}{\sqrt t}e^{-ct^{3/2}/h}dt=\ord(\eps e^{-c\lambda^{3/2}}),
$$
where $c>0$ is a constant, and the estimate is uniform with respect to $\eps>0$ small enough and $\lambda >1$ large enough such that $\lambda h^{2/3}\to 0$. In particular, taking,
\be
\label{choicelambda}
\lambda \geq (c^{-1}|\ln\eps |)^{2/3},
\ee
we obtain,
$$
\alpha_{1,L}^-(\lambda)=\ord(\eps^2).
$$
On the other hand, using Propositions \ref{u1Losc}, \ref{yaf1} and \ref{yaf2}, we obtain,
$$
\begin{aligned}
\alpha_{1,L}^+(\lambda)= &\frac{-2\pi}{\eps}(\sin\frac{\mathcal A}{h})\int_{-\lambda h^{\frac23}}^0r_0(t)(\xi_1'\xi_2')^{-\frac12}\ai (h^{-\frac23}\xi_1)\check\ai (h^{-\frac23}\xi_2)dt\\
& -\frac{2\pi}{\eps}(\cos\frac{\mathcal A}{h})\int_{-\lambda h^{\frac23}}^0r_0(t)(\xi_1'\xi_2')^{-\frac12}\bi (h^{-\frac23}\xi_1)\check\ai (h^{-\frac23}\xi_2)dt +\ord(\eps^2).
\end{aligned}
$$
Making the change of variable $y:=\eps^{-2}t$, and using that, for $y\in [-\lambda, 0]$, we have  $r_0(\eps^2y)=r_0(0)+\ord (\eps^2\lambda)$, $\xi_j'(\eps^2 y)=1+\ord(\eps^2\lambda)$, $\ai(\eps^{-2}\xi_1(\eps^2y))=\ai (y-\rho)+\ord(\lambda^2\eps^2)$, $\check\ai(\eps^{-2}\xi_2(\eps^2y))=\check\ai (y+\rho)+\ord(\lambda^2\eps^2)$, $\bi(\eps^{-2}\xi_1(\eps^2y))=\bi (y-\rho)+\ord(\lambda^2\eps^2)$, we obtain,
$$
\begin{aligned}
\alpha_{1,L}^+(\lambda)= & -2\eps \pi r_0(0)(\sin\frac{\mathcal A}{h})\int_{-\lambda}^0\ai (y-\rho)\check\ai (y+\rho)dy\\
& -2\eps\pi r_0(0)(\cos\frac{\mathcal A}{h})\int_{-\lambda}^0\bi (y-\rho)\check\ai (y+\rho)dy +\ord(\eps^2)+\ord(\lambda^3\eps^3).
\end{aligned}
$$
Then, using the behaviour of $\ai$, $\check\ai$ and $\bi$ at $-\infty$, this leads to,
$$
\begin{aligned}
& \alpha_{1,L}^+(\lambda)\\
& =- 2\eps \pi r_0(0)\int_{-\infty}^0\check\ai (y+\rho)\left((\sin\frac{\mathcal A}{h})\ai (y-\rho)+(\cos\frac{\mathcal A}{h})\bi (y-\rho)\right)dy\\
&\quad  +\ord(\eps e^{-c'\lambda^{2/3}})+\ord(\eps^2)+\ord(\lambda^3\eps^3),
\end{aligned}
$$
with $c'>0$ constant. Therefore, taking $\lambda:=((c'')^{-1}|\ln\eps |)^{2/3}$ with $c=\min\{c,c'\}$, and using the fact that $\im\,\rho=\ord(\eps)$, we obtain the required approximation of $\alpha_{1,L}$. The approximations of $\alpha_{2,L}$, $\alpha_{1,R}$ and $\alpha_{2,R}$ are obtained in a very similar way (starting from (\ref{defconstjS})), and we omit the proofs.

Now we study $\beta_{1,R}$. We have,
$$
\begin{aligned}
\beta_{1,R} =  &\frac{\pi^2}{4\eps^2}\int_0^{+\infty} u_{1,R}^-(t)W(t,hD_t)\left[ u_{2,R}^-(t)\int_0^t u_{2,R}^+(s)(W^*u_{1,R}^-)(s)ds\right.\\
& \left.+u_{2,R}^+(t)(W^*u_{1,R}^-)(s)ds\int_t^{+\infty}u_{2,R}^-(s)(W^*u_{1,R}^-)(s)ds\right] dt +\ord(h),
\end{aligned}
$$
where the integrals run over $I_R^\theta$. Because of the exponential decay of $u_{1,R}^-$ away from 0, we immediately see that only a neighbourhood of 0 contributes to the integrals (up to $\ord(e^{-c/h})$ with $c>0$). Moreover, making an integration by parts and using the behaviour of all the functions near 0, we see as for $\alpha_{1,L}$ that, up to an error $\ord(\eps^3)$, $W$ and $W^*$ can be replaced by $r_0$. Finally, using (\ref{defu2Rpm}) and the expressions of $u_{1,R}^-$, $u_{2,L}^-$ and $u_{2,L}^+$ given in Propositions \ref{yaf1} and \ref{yaf2}, and proceeding as for $\alpha_{1,L}$, we obtain,
$$
\begin{aligned}
\beta_{1,R} &= 2i\pi^2r_0(0)^2\eps^2\iint_{0\leq s\leq t}\! \ai (t-\rho)\ai(s-\rho)\check {\rm A}_{out}(t+\rho)\check {\rm A}_{in}(s+\rho)dsdt\\
& +2i\pi^2r_0(0)^2\eps^2\iint_{0\leq t\leq s}\ai (t-\rho)\! \ai(s-\rho)\check {\rm A}_{in}(t+\rho)\check {\rm A}_{out}(s+\rho)dsdt\\
& + \ord (h),
\end{aligned}
$$
where we have set,
$$
\check {\rm A}_{out}:= \check\ai -i\check\bi \quad ; \quad \check {\rm A}_{in}:= \check\ai +i\check\bi.
$$
Exchanging $t$ and $s$ in the second integral, and replacing $\check {\rm A}_{out}$ and $\check {\rm A}_{in}$ by their expressions, an elementary computation (plus the fact that $\im\,\rho=\ord(\eps)$) leads to the required approximation of $\beta_{1,R}$.

The same procedure (but somehow simpler) shows that $\beta_{1,L}$, $\beta_{2,L}$ and $\beta_{2,R}$ are $\ord (\eps^2)$ (note that, for $\beta_{1,L}$ and $\beta_{2,R}$, one must use the fact that the integral $\int_0^\infty\int_0^\infty (ts)^{-{\frac12}}e^{-|t^{\frac32} -s^{\frac32}|}dtds$ is finite).
Finally, concerning $\beta_{1,L}$, we see on (\ref{defbeta1L}) that it involves only functions that are real when $E$ is real, and the same kind of estimates as for $\beta_{1,R}$ show that it is real up to $\ord (h)$.
\end{proof}

\section{Quantization condition}
\label{sect6}

In order to simplify the writing, we will use the following notations: 

If $u_1$ is any of the functions $u_{1,L}^\pm$ or $u_{1,R}^\pm$, we set,
$$
{\mathbf u_1}:=\left[ 
\begin{array}{c}
u_1(0) \\
\widetilde\partial u_1(0)\\
0\\
0
\end{array}
\right],
$$
and if $u_2$ is any of the functions $u_{2,L}^\pm$ or $u_{2,R}^\pm$, we set,
$$
{\mathbf u_2}:=\left[ 
\begin{array}{c}
0 \\
0\\
u_2(0)\\
\widetilde\partial u_2(0)
\end{array}
\right].
$$
With these notations, we see on Proposition \ref{calculen0} that we have,
\be
\W_0(E):=\widetilde\W (w_{1,L}, w_{2,L}, w_{1,R}, w_{2,R}) =\det ({\mathbf v}_{1,L},{\mathbf v}_{1,R}, {\mathbf v}_{2,L}, {\mathbf v}_{2,R}) +\ord (h),
\ee
with,
$$
\begin{aligned}
&{\mathbf v}_{1,S}:={\mathbf u_{1,S}^-}+\alpha_{1,S}{\mathbf u_{2,S}^+} +\beta_{1,S}{\mathbf u_{1,S}^+};\\
&{\mathbf v}_{2,S}:={\mathbf u_{2,S}^-}+\alpha_{2,S}{\mathbf u_{1,S}^+} +\beta_{2,S}{\mathbf u_{2,S}^+}\quad (S=L,R).
\end{aligned}
$$
Developing the determinant by multi-linearity, and observing that all the terms that involve at least three vectors with the same value of the index $j$ vanish, we obtain,
$$
\begin{aligned}
\W_0(E)=&\widetilde\W(u_{1,L}^-, u_{1,R}^-)\widetilde\W(u_{2,L}^-, u_{2,R}^-)\\
& + \beta_{2,R}\widetilde\W (u_{1,L}^-,u_{1,R}^-)\widetilde\W(u_{2,L}^-, u_{2,R}^+)\\
& +\beta_{2,L}\widetilde\W (u_{1,L}^-,u_{1,R}^-)\widetilde\W(u_{2,L}^+, u_{2,R}^-)\\
& + \alpha_{1,R}\alpha_{2,R}\widetilde\W (u_{1,L}^-, u_{1,R}^+)\widetilde\W(u_{2,R}^+, u_{2,L}^-)\\
& +\alpha_{1,R}\alpha_{2,L}\widetilde\W(u_{1,L}^-,u_{1,L}^+)\widetilde\W (u_{2,R}^-, u_{2,R}^+)\\
& + \beta_{1,R}\widetilde\W(u_{1,L}^-, u_{1,R}^+)\widetilde\W(u_{2,L}^-,u_{2,R}^-)\\
& +\alpha_{1,L}\alpha_{2,R}\widetilde\W(u_{2,L}^+,u_{2,L}^-)\widetilde\W(u_{1,R}^+, u_{1,R}^-)\\
& +\alpha_{1,L}\alpha_{2,L}\widetilde\W (u_{2,L}^+,u_{2,R}^-)\widetilde\W(u_{1,R}^-, u_{1,L}^+)\\
& +\beta_{1,L}\widetilde\W(u_{1,L}^+,u_{1,R}^-)\widetilde\W(u_{2,L}^-,u_{2,R}^-) +\ord(h).
\end{aligned}
$$

In particular, we observe that each $\alpha_{j,S}$ is always multiplied by another similar quantity, that is, by $\ord(\eps)$. As a consequence, an error on $\alpha_{j,S}$ of order $\eps^2$ will lead to a error of order $h$ in $\W_0(E)$, and thus, by Proposition \ref{calculen0suite}, we can replace $\alpha_{2,S}$ by $\alpha_{1,S}$ ($S=L,R$).
Then, computing the various Wronskians  that appear (see Appendix \ref{appendix3}), we find,
\be
\label{W0modh}
\begin{aligned}
-i\pi^2e^{-i\frac{\pi}4} \W_0(E)= &-4\sqrt{2} (\cos\frac{{\mathcal A}}{h})(1+\ord(\eps^2))
+4\sqrt{2}(\sin\frac{\mathcal A}{h})\alpha_{1,R}^2\\
& +8e^{\frac{i{\pi}}4}\alpha_{1,R}\alpha_{1,L}+2\sqrt{2}(\sin\frac{\mathcal A}{h})\beta_{1,R}+i\sqrt{2}(\sin\frac{\mathcal A}{h})\alpha_{1,L}^2\\
& +2\sqrt{2}(\sin\frac{{\mathcal A}}{h})\beta_{1,L}+\ord (h).
\end{aligned}
\ee
Finally, observing that we have,
$$
\begin{aligned}
& \alpha_{1,R}^2 = \frac{\pi^2 r_0(0)^2\eps^2}2 (2\mu_A\mu_B+i\mu_A^2-i\mu_B^2)+\ord(h);\\
& \alpha_{1,L}\alpha_{1,R}=\pi^2 r_0(0)^2\eps^2\sqrt{2}e^{i\frac{\pi}4}(\sin\frac{\mathcal A}{h})\mu_A(\mu_A-i\mu_B)+\ord(h+\eps^2|\cos\frac{\mathcal A}{h}|);\\
& \alpha_{1,L}^2=4\pi^2 r_0(0)^2\eps^2(\sin\frac{\mathcal A}{h})^2\mu_A^2+\ord(h+\eps^2|\cos\frac{\mathcal A}{h}|);\\
& \beta_{1,R}=\re\,\beta_{1,R}+i\pi^2 r_0(0)^2\eps^2(\mu_A^2+\mu_B^2)+\ord (h);\\
& \beta_{1,L} =\re\,\beta_{1,L}+\ord (h),
\end{aligned}
$$
we obtain,
\be
\label{W0modh1}
\begin{aligned}
-i\pi^2e^{-i\frac{\pi}4} \W_0(E)= &-4\sqrt{2} (\cos\frac{{\mathcal A}}{h})(1+\ord(\eps^2))+2\sqrt{2}(\sin\frac{\mathcal A}{h} )\re\, (\beta_{1,L}+\beta_{1,R})\\
& +4\sqrt{2}(\pi r_0(0)\eps)^2(\sin\frac{\mathcal A}{h} )\left(3\mu_A\mu_B +i(3+\sin^2\frac{\mathcal A}{h})\mu_A^2\right)\\
&  +\ord (h).
\end{aligned}
\ee

Now, by Proposition \ref{quantcondi1}, the quantization condition reads,
$$
\W_0(E)=0.
$$
Hence, in view of (\ref{W0modh}) and of Proposition \ref{calculen0}, if we set,
$$
b_0:=\frac1{2\eps^2}\re ( \beta_{1,L}+\beta_{1,R}) +3\pi^2r_0(0)^2\mu_A(\re\, \rho)\mu_B(\re\,\rho)=\ord(1),
$$
we have proved,
\begin{proposition}\sl
$E=\rho\eps^2 \in {\mathcal D_h}(C_0)$ is a resonance of $P$ if and only if,
\begin{equation}
\label{BS}
\cos \frac{{\mathcal A}(E)}{\e^3}  
=\e^2 F(E,\e),
\end{equation}
where
\begin{equation}
\label{BSR}
F(E,\e)= \left(b_0 +i\pi^2r_0(0)^2(3+\sin^2\frac{{\mathcal A}(E)}{\e^3})\mu_A(\re\,\rho)^2\right)\sin\frac{{\mathcal A}(E)}{\e^3}
+\ord (\e).
\end{equation}
\end{proposition}

\begin{rem}\sl
The quantization condition \eq{BS} is 
of Bohr-Sommerfeld type associated with the
single potential well of $V_1(x)$.
The imaginary part of $F(E,\e)$ will give an estimate
on the width of resonances.
\end{rem}

\section{Completion of the proof}
\label{sect7}
In order to solve (\ref{BS}) in ${\mathcal D}_h(C_0)$ (where $C_0$ may actually vary a little bit in order to avoid ''border'' effects), we first observe that, near $E=0$, the roots of the equation $\cos({\mathcal A}(E)/h)=0$ are given by $E=e_k(h)$ with,
$$
e_k(h):={\mathcal A}^{-1}\left((k+\frac12)\pi h\right) \in\R \quad (k\in\Z).
$$
(Here, $k\in\Z$ must be taken such a way that ${\mathcal A}^{-1}((k+\frac12)\pi h)$ is effectively close to 0.) In particular, restricting to $E=\ord(\eps^2)$, and writing ${\mathcal A}(E) = {\mathcal A}(0) + E{\mathcal A}'(0) +\frac12E^2{\mathcal A}''(0)+\ord (E^3)$, we obtain the well known relation,
$$
e_k(h)=\lambda_k(h)\eps^2-\frac{\lambda_k(h)^2{\mathcal A}''(0)}{{2\mathcal A}'(0)}\eps^4  + \ord(\eps^6),
$$
where $\lambda_k(h)$ is defined in (\ref{deflambdakh}). In particular, the distance between two consecutive $e_k(h)$'s is of order $h$.  Moreover, if $E\in\C$ is such that $E=\ord (\eps^2)$ and ${\mathcal A}(E)$ stays at a distance greater than $\delta h$ from the set $\{ e_k(h)\, ;\, k\in\Z\}$, with $\delta >0$ constant, then $\cos({\mathcal A}(E)/h)$ remains at some fix positive distance from 0. As a consequence, for $\eps >0$ small enough, we can apply the Rouch\'e theorem and conclude that, for each $k$ such that $\lambda_k(h)=\ord (1)$, there exists a unique solution $E_k(h)$ to (\ref{BS}) such that,
$$
E_k(h) =e_k(h)+ o(h),
$$
and, conversely, all the roots of (\ref{BS}) in ${\mathcal D}_h(C_0)$ are of this type. 

Now, going back to equation (\ref{BS}), we immediately see that, actually, we have,
$$
E_k(h) =e_k(h)+ \ord (\eps^5),
$$
so that (\ref{reEk}) is proved in the case $\tau_1=\tau_2=1$

In order to prove (\ref{imEk}), we first observe that, by (\ref{reEk}), the equation (\ref{BS}) implies,
$$
\cos \frac{{\mathcal A}(E_k(h))}{\e^3}  
=\e^2 F(\lambda_k(h)\eps^2,\e)+\ord (\eps^3).
$$
Then, taking the local inverse of $\cos$ near $(k+\frac12)\pi$, and the inverse of $\mathcal A$ near $E=0$, (\ref{imEk}) (with $\tau_1=\tau_2=1$) immediately follows from (\ref{BSR}) and the fact that $\sin \frac{{\mathcal A}(\lambda_k(h)\eps^2)}{\e^3}  =(-1)^k +\ord(\eps^2)$.

When $\tau_1$ and $\tau_2$ are general positive numbers, we observe that all the constructions Sections \ref{sect3} to \ref{sect5} and of Appendix \ref{appendix2} remain completely unchanged. Therefore, the proof proceeds exactly in the same way, and the only differences are in the approximate values of $\xi_1(\eps^{-2}y)$ and $\xi_2(\eps^{-2}y)$. A very simple computation shows that they become,
$$
\begin{aligned}
\xi_1(\eps^{-2}y)=\tau_1^{\frac13}\left(y-\frac{\rho}{\tau_1}\right) +\ord(\eps^2);\\
\xi_2(\eps^{-2}y)=\tau_2^{\frac13}\left(y+\frac{\rho}{\tau_2}\right) +\ord(\eps^2).
\end{aligned}
$$
As a consequence, the approximations given in (\ref{approxalphabeta}) have to be changed, too, and indeed now the functions $\mu_A$ and $ \mu_B$ will depend also on the side where we are working (on $I_L$ or on $I_R$). On $I_R$ (that is, for $\alpha_{1,R}$ and $\beta_{1,R}$) they will just be as before, with $\ai (y-\rho), \bi(y-\rho)$ substituted by $\ai(\tau_1^{-2/3}(\tau_1y-\rho)), \bi(\tau_1^{-2/3}(\tau_1y-\rho))$,  and $\check\ai (y+\rho),\check \bi(y+\rho)$ substituted by $\check\ai(\tau_2^{-2/3}(\tau_2y+\rho), \check\bi(\tau_2^{-2/3}(\tau_2y+\rho)$. On $I_L$ (that is, for $\alpha_{1,L}$), $\mu_A$ and $ \mu_B$ become,
$$
\begin{aligned}
\widetilde\mu_A(t)=\int_{-\infty}^0\ai(\tau_1^{-2/3}(\tau_1y-\rho))\check\ai(\tau_2^{-2/3}(\tau_2y+\rho))dy ;\\
\widetilde\mu_B(t)=\int_{-\infty}^0\bi(\tau_1^{-2/3}(\tau_1y-\rho))\check\ai(\tau_2^{-2/3}(\tau_2y+\rho))dy.
\end{aligned}
$$
But the computations proceeds in a similar way, and since $\widetilde\mu_A(t)=\mu_2(t)$, the result follows in the general case, too.

\begin{remark}\sl {\bf Resonant states} By construction, the resonant state $\varphi_k$ associated with $E_k(h)$ can be written both as a linear combination of $w_{1,L}$ and $w_{2,L}$, and  a linear combination of $w_{1,R}$ and $w_{2,R}$ (all computed at $E=E_k(h)$). The coefficients can actually be computed (up to $\ord (h)$) by using Proposition \ref{estw1S}, identifying each $w_{j,S}$ with the  vector $\vec{\mathbf w}_{j,L}$ of $\R^4$  given by,
$$
\vec{\mathbf w}_{j,L}:= \left[\begin{array}{c}
w_{j,S}^-(0)\\
\widetilde\partial w_{j,S}^-(0)
\end{array}\right].
$$
\end{remark}
Then, the approximations of the various functions involved with their Airy representation (obtained from Propositions \ref{yaf1} and \ref{yaf2}) permit us to write $\varphi_k$ as,
\be
\varphi_k = w_{1,L}-\mu w_{2,L} = \lambda w_{1,R}+\nu w_{2,R},
\ee
with,
$$
\begin{aligned}
&  \lambda = \sin\frac{{\mathcal A}}{h}+\ord (h)=(-1)^k +\ord (\eps^2);\\
& \nu = -2\lambda \alpha_{1,R} + \ord (h)= 4\eps \pi r_0(0) \mu_A(\lambda_k(h)) +\ord (h);\\
& \mu = \cos \frac{{\mathcal A}}{h}-\frac12(\beta_{1,L}+\lambda \beta_{1,R}+\nu\alpha_{2,R})+\ord (h) =\ord (\eps^2).
\end{aligned}
$$
Using the various asymptotic behaviours of the functions $w_{j,S}$'s, one can derive the (semiclassical) asymptotic behaviour of $\varphi_k$ in all of $\R$.

\section{Appendix}
\label{sect8}
\subsection{Inhomogeneous Airy equations}
\label{appendix1}

Let ${\rm Ai}\,(y)$ be the Airy function, which is characterised
as the solution to the Airy equation
\begin{equation}
\label{airy}
-u''(y)+yu(y)=0
\end{equation}
with exponential decay as $y$ tends to $+\infty$,
$$
{\rm Ai}\,(y) = \frac{1}{2\sqrt{\pi}}\, y^{-\frac{1}{4}}\, e^{-\frac{2}{3}y^{3/2}}(1+\ord(y^{-1}))
\qquad (y\to +\infty).
$$
It is oscillating when $y<0$, and its behaviour is, 
$$
{\rm Ai}\,(y) = \frac{1}{\sqrt{\pi}} \, (-y)^{-\frac{1}{4}} \sin \left(\frac{2}{3}(-y)^{3/2} + \frac{\pi}{4}\right)(1+\ord(|y|^{-1}))
\qquad {\rm as} \quad y\to -\infty.
$$
Moreover, the asymptotic behaviour of its derivative $\ai'(y)$ is obtained by formally differentiating the previous ones, and these asymptotic behaviours remain valid in sufficiently small complex sectors around the real line.

We define another solution ${\rm Bi}\,(y)$ to the Airy equation
by the asymptotic behavior as $y\to -\infty$
$$
{\rm Bi}\,(y) = \frac{-1}{\sqrt{\pi}} \, (-y)^{-\frac{1}{4}} \sin \left(\frac{2}{3}(-y)^{3/2} - \frac{\pi}{4}\right)(1+\ord(|y|^{-1}))
\qquad (y\to -\infty).
$$
${\rm Bi}\,(y)$ is positive and grows exponentially for $y>0$,
and satisfies
$$
{\rm Bi}\,(y) = \frac{1}{\sqrt{\pi}}\, y^{-\frac{1}{4}}\, e^{\frac{2}{3}y^{3/2}}(1+\ord(y^{-1}))
\qquad (y\to +\infty).
$$
From the asymptotic behaviors of ${\rm Ai}\,(y)$ and
${\rm Bi}\, (y)$ as $y\to -\infty$, we easily see the following properties.
At first, the solutions
\begin{equation}\label{rightgoingAiry}
{\rm Ai}\,(y)-i{\rm Bi}\,(y)\sim 
\frac{e^{-\pi i/4}}{\sqrt{\pi}} \, (-y)^{-\frac{1}{4}} 
\exp \left(\frac{2i}{3}(-y)^{3/2}\right)\qquad (y\to -\infty);
\end{equation}
\begin{equation}\label{leftgoingAiry}
{\rm Ai}\,(y)+i{\rm Bi}\,(y)\sim 
\frac{e^{\pi i/4}}{\sqrt{\pi}} \, (-y)^{-\frac{1}{4}} 
\exp \left(-\frac{2i}{3}(-y)^{3/2}\right)\qquad (y\to -\infty),
\end{equation}
are outgoing and incoming respectively
for negative $y$, and secondly, the wronskian of
${\rm Ai}\,(y)$ and
${\rm Bi}\, (y)$ is given by
$${\mathcal W}[{\rm Ai},{\rm Bi}] := 
{\rm Ai}\, (y){\rm Bi}'\, (y)
-{\rm Ai}'\, (y){\rm Bi}\, (y)=\pi^{-1}.$$

Set
$$
K(y,z):=-\pi \left\{\ai (y)\bi (z)
-\ai (z)\bi (y)\right\},
$$
and define the integral operators $\K$ 
and $\widetilde{\K}$ 
for $f\in C_0^\infty(\R)$ 
\be
\label{defKKtilde}
\K [f](y):=\int_y^0 K(y,z)f(z)dz\, ;\,
\widetilde \K [f](y):=\int^y_0 K(-y,-z)f(z)dz.
\ee
The function $\K [f](y)$ gives a particular solution to the inhomogeneous equation
$$
-u''+yu=f,
$$
while $\widetilde \K [f](y)$ gives
 a particular
solution to
$$
-u''-yu=f.
$$
Notice that there exists a symmetric property between $\K$ 
and $\widetilde{\K}$, namely,
\begin{equation}\label{symmetric property} 
\K[f](-y) 
= \widetilde{\K}[\check{f}](y),
\end{equation}
where $\check f(y)=f(-y)$. Moreover, if $\rho\in\C$, then the operators ${\K}_\rho$ and $ \widetilde{\K}_\rho$ defined by,
\be
\label{defKrho}
{\mathcal K}_\rho[f] (y):= \int_y^0 K(y-\rho, z-\rho)f(z)dz\,;\,\widetilde{\mathcal K}_\rho[f] (y):= \int_0^y K(-y-\rho, -z-\rho)f(z)dz
\ee
give solutions to the equations
$$
-u''+(y-\rho)u=f\quad ;\quad -u''-(y+\rho)u =f,
$$
respectively.
\begin{remark}\sl Observe that all these constructions remain valid when $y$ becomes complex, as long as $\im\, y$ stays bounded. In that case, the integrals must be taken along any complex curve joining $y$ to 0.
\end{remark}

\subsection{Yafaev's constructions}
\label{appendix2}

In this appendix, we recall and extend the constructions made in \cite{Ya} for the scalar Schr\"odinger equation $(P_j-E)u=0$. In \cite{Ya} such constructions are made for real $E$ only, and they just concern solutions decaying at infinity. Here, we need to consider complex values of $E$ and exponentially large or oscillating solutions, too. 

We fix $x_0\in (x^*,0)$, and we first treat the case $j=1$. For $E$ small enough, we denote by $x_1=x_1(E)$ the only point near $0$ where $V_1-E$ vanishes (in particular, $x_1(E)=E+\ord(E^2)$ depends analytically on $E$). In the particular case where $E$ is real, we can define as in \cite{Ya},
\be
\begin{aligned}
& \xi_1(x;E) =\left(\frac32 \int_{x_1(E)}^x \sqrt{V_1(t)-E}\, dt\right)^{\frac23} \quad \mbox{when } x\geq x_1;\\
& \xi_1(x;E) =-\left(\frac32 \int_x^{x_1(E)} \sqrt{E-V_1(t)}\, dt\right)^{\frac23} \quad \mbox{when } x_0\leq x\leq x_1.
\end{aligned}
\ee
Then, it is easy to check that $\xi_1(x;E) = x-x_1(E) + \ord((x-x_1)^2)$ as $x\to x_1$, and that $\xi_1$ depends analytically on $x$ and $E$ for $x\in (x_0,+\infty)$, $E$ small enough. Since also $V_1$ has a positive limit at $+\infty$, we see that we can extend analytically $\xi_1$ to a complex neighbourhood of $(x_0,+\infty)\times\{0\}$. Then, $\xi_1(x)$ satisfies $(\xi_1')^2\xi_1 = V_1-E$ and $\re\, \xi_1' >0$ everywhere on $[x_0,+\infty)$. In particular, when $E\in{\mathcal D}_h(C_0)$ is fixed and $x$ varies in $(x_0,+\infty)$, then $\xi_1$ describes a smooth complex curve parametrised by $x$, with $\im \xi_1(x) =\ord (h)$ uniformly. From now on, in order to simplify the notations, we drop the dependance of $\xi_1$ with respect to $E$.
The result is (see also \cite{Ya}, Theorem 2.5),

\begin{proposition}\sl
\label{yaf1}
 Let $E\in{\mathcal D}_h(C_0)$. Then, the equation $(P_1 -E)u=0$ admits two solutions $u_{1,R}^\pm$ on $\R$, such that, as $x\to +\infty$, (and uniformly with respect to $h>0$ small enough),
 $$
u_{1,R}^\pm(x) \sim (1+\ord (h))\frac{h^{\frac16}}{\sqrt\pi}(V_1(x) - E)^{-1/4} e^{ \pm\int_{x_1(E)}^x\sqrt{V_1(t) - E}\, dt/h},
$$
 and, as $h\to 0_+$,
$$
\begin{aligned}
& u_{1,R}^- (x) =2(\xi_1'(x))^{-\frac12}\ai (h^{-\frac23}\xi_1(x))(1+\ord(h))\\
& \hskip 7cm \mbox{  on } [x_0, +\infty)\cap\{\re\,\xi_1(x) \geq0\};\\
& u_{1,R}^- (x) =2(\xi_1'(x))^{-\frac12}\ai (h^{-\frac23}\xi_1(x))+\ord(h (1+h^{-2/3}|\xi_1(x)|)^{-\frac14})) \\
& \hskip 7cm\mbox{  on } [x_0, +\infty)\cap\{\re\,\xi_1(x) \leq0\};\\
& u_{1,R}^+ (x) =(\xi_1'(x))^{-\frac12}\bi (h^{-\frac23}\xi_1(x))(1+\ord(h))\\
& \hskip 7cm \mbox{  on } [x_0, +\infty)\cap\{\re\,\xi_1(x) \geq0\};\\
& u_{1,R}^+ (x) =(\xi_1'(x))^{-\frac12}\bi (h^{-\frac23}\xi_1(x))+\ord(h (1+h^{-2/3}|\xi_1(x)|)^{-\frac14}))\\
& \hskip 7cm \mbox{  on } [x_0, +\infty)\cap\{\re\,\xi_1(x) \leq0\}.
\end{aligned}
$$
\end{proposition}
\begin{proof}
The proof for $u_{1,R}^- $ is the same as in \cite{Ya}, with the difference that, here, $E$ may be complex. However, since we have $\im E=\ord (h)$, all the estimates in \cite{Ya} remain valid. Observe, in particular, that when $\re\, \xi_1(x)\geq 0$, then $\im\, (\xi_1(x))^{3/2}=\ord (h)$, and thus $\ai (h^{-\frac23}\xi_1(x))\not= 0$.

Therefore, let us focus on the construction of $u_{1,R}^+$. As in \cite{Ya}, Section 3, setting $t:=h^{-2/3}\xi_1(x)$ and $f( t):=\xi_1'(x)^{\frac12}u(x)$, the equation $(P_1-E)u=0$ becomes,
\be
\label{eqfN}
-f''(t) + tf(t)=R(t)f(t),
\ee
with, 
\be
\begin{aligned}
& R(t) =h^{4/3}p(h^{2/3}t);\\
& p(x):=\left[(\xi_1'(x))^{-1/2}\right]''(\xi_1'(x))^{-3/2}=\ord (1+|\xi_1(x)|)^{-2}.
\end{aligned}
\ee
(In the last estimate, we have used the fact that $|V_1'(x)|^2+|V_1''(x)|=\ord((1+|x|)^{-2})$.) Defining $\K$ as in (\ref{defKKtilde}), we reduce (\ref{eqfN}) to the Volterra equation,
\be
\label{volt}
f =\bi + \K[Rf].
\ee
Then, a continuous solution of (\ref{volt}) will be solution of (\ref{eqfN}), too, and we expect it to have the right behaviour at infinity. Moreover, it is enough to solve (\ref{volt}) separately on $\re\, t\geq 0$ and $\re \, t\leq 0$ (where, in any case, $t$ remains on the curve
$\Gamma:=\{ h^{-2/3}\xi_1(x)\, ;\, x\in[x_0,+\infty)\}$).

When $\re\, t\geq 0$, one has $\bi (t)\not =0$, and we set $g:=f/\bi$. Then $g$ must be solution to,
\be
\label{voltred}
g=1+Lg,
\ee
with,
$$
Lg(t):=\pi\int_0^t \left(\frac{\ai(t)}{\bi(t)}\bi(s)^2-\ai(s)\bi(s)\right)R(s)g(s)ds.
$$
Using the asymptotic behaviours of $\ai$ and $\bi$ at infinity, and the fact that $\int_0^\infty (1+s)^{-\frac12}(1+h^{\frac23}s)^{-2}ds =\ord (h^{-\frac13})$, we see that $||L||_{C_b^0(\Gamma\cap\{\re\, t\geq 0\})} =\ord (h)$. Therefore, (\ref{voltred}) can be solved by iteration on $\Gamma_+:=\Gamma\cap\{\re\, t\geq 0\}$, and the corresponding solution to (\ref{volt}) satisfies,
$$
f=\bi(t)(1+\ord (h))\, \mbox{ uniformly}.
$$
On $\Gamma_-:=\Gamma\cap\{\re\, t\leq 0\}$, since $t =\ord(h^{-2/3})$ there, and $|\ai (s)| +|\bi (s)| =\ord ((|1+|s|)^{-1/4})$ between 0 and $t$, we obtain 
$$
(1+|t|)^{\frac14}|\K[Rf](t)|=\ord (h)\sup_{s\in\Gamma_-}(1+|s|)^{\frac14}|f(s)| ,
$$
and thus, (\ref{volt}) can be solved by iteration there, leading to a solution that satisfies,
$$
f=\bi(t)+\ord (h(1+|t|)^{-\frac14})\, \mbox{ uniformly}.
$$
Moreover, the behaviour at infinity of $f$ is obtained from that of $\bi$ and of $\xi_1(x)$. This completes the proof of the proposition.
\end{proof}
\begin{remark}\sl
\label{derivees}
We also see on (\ref{voltred}) that $g' = \ord (h(1+|t|)^{1/2})$. This leads to,
$$
\begin{aligned}
h^{2/3}(u_{1,R}^+)' (x) = &(\xi_1'(x))^{\frac12}\bi' (h^{-\frac23}\xi_1(x))  (1+\ord(h))\\
&+\ord(h^{2/3})(h^{2/3}+|\xi_1(x)|)^{1/2}(\xi_1'(x))^{-\frac12}\bi (h^{-\frac23}\xi_1(x))
\end{aligned}
$$
on $[x_0, +\infty)\cap\{\re\,\xi_1(x) \geq0\}$, and similar estimates are valid on $[x_0, +\infty)\cap\{\re\,\xi_1(x) \leq0\}$, and for $h^{2/3}(u_{1,R}^-)'$, too. For instance, on $[x_0, +\infty)\cap\{\re\,\xi_1(x) \leq0\}$, one obtains,
\be
\label{uRprime}
\begin{aligned}
h^{2/3}(u_{1,R}^+)' (x) = (\xi_1'(x))^{\frac12}& \bi' (h^{-\frac23}\xi_1(x))  \\
&+\ord(h^{5/6})(1+(h^{2/3}+|\xi_1(x)|)^{-1/4});\\
h^{2/3}(u_{1,R}^-)' (x) =2 (\xi_1'(x))^{\frac12}& \ai' (h^{-\frac23}\xi_1(x))  \\
&+\ord(h^{5/6})(1+(h^{2/3}+|\xi_1(x)|)^{-1/4}).
\end{aligned}
\ee
\end{remark}
\begin{remark}\sl 
Similar constructions can be done on $(-\infty, x_0]$, leading to solutions $u_{1,L}^\pm$ with the asymptotic behaviour,
$$
u_{1,L}^\pm(x) \sim \frac{h^{\frac16}}{\sqrt\pi}(V_1(x) - E)^{-1/4} e^{\mp \int_{x^*(E)}^x\sqrt{V_1(t) - E}\, dt/h}\qquad (x \to -\infty),
$$
where $x^*(E)$ is the only point near $x^*$ where $V_1-E$ vanishes.
\end{remark}

Now, we treat the case $j=2$. Here the situation is a bit different, because the set where $V_2< 0$ is unbounded, and also because there is one turning point only. This actually permits us to directly obtain the asymptotic of the solutions both at $-\infty$ and at $+\infty$. We denote by $x_2(E)$ the unique point near 0 where $V_2-E$ vanishes, and, when $E$ is real, we set,
\be
\begin{aligned}
\label{defxi2}
& \xi_2(x;E)) =\left(\frac32 \int_{x_2(E)}^x \sqrt{E-V_2(t)}\, dt\right)^{\frac23} \quad \mbox{when } x\geq x_2;\\
& \xi_2(x;E) =-\left(\frac32 \int_x^{x_2(E)} \sqrt{V_2(t)-E}\, dt\right)^{\frac23} \quad \mbox{when } x\leq x_2.
\end{aligned}
\ee
As before,  we extend analytically this definition to complex values of $E$, and
we have,
\begin{proposition}\sl 
\label{yaf2}
Let $E\in{\mathcal D}_h(C_0)$. Then, there exist two solutions $u_{2,L}^\pm$ to equation $(P_2 -E)u=0$ on $\R$, and two constants $a_2^\pm = 1+\ord (h)$, such that, 
\be
 u_{2,L}^\pm(x) \sim (1+\ord (h))\frac{h^{\frac16}}{\sqrt\pi}(V_2(x)-E )^{-1/4} e^{\mp\int_{x_2(E)}^x\sqrt{V_2(t)-E}\, dt/h}\qquad (x \to -\infty);
\ee
 \be
 \label{asympsoolsort}
 \begin{aligned}
\frac12a_2^-& u_{2,L}^-(x) \pm ia_2^+u_{2,L}^+ (x)\\
& \sim (1+\ord (h))\frac{h^{\frac16}}{\sqrt\pi}(E-V_2(x) )^{-1/4} e^{\mp i \int_{x_2(E)}^x\sqrt{E-V_2(t)}\, dt/h}
\quad  (x \to +\infty),
\end{aligned}
\ee
uniformly with respect to $h>0$ small enough,
and, 
$$
\begin{aligned}
& u_{2,L}^- (x) =2(\xi_2'(x))^{-\frac12}\check\ai (h^{-\frac23}\xi_2(x))+\ord(h (1+h^{-2/3}|\xi_2(x)|)^{-\frac14}))\\
& \hskip 8cm  \mbox{  on } \R\cap\{\re\,\xi_2(x) \geq0\};\\
& u_{2,L}^- (x) =2(\xi_2'(x))^{-\frac12}\check\ai (h^{-\frac23}\xi_2(x))(1+\ord(h)) \mbox{  on } \R\cap\{\re\,\xi_2(x) \leq0\};\\
& u_{2,L}^+ (x) =(\xi_2'(x))^{-\frac12}\check\bi (h^{-\frac23}\xi_2(x))+\ord(h (1+h^{-2/3}|\xi_2(x)|)^{-\frac14}))\\
& \hskip 8cm  \mbox{  on } \R\cap\{\re\,\xi_2(x) \geq0\};\\
& u_{2,L}^+ (x) =(\xi_2'(x))^{-\frac12}\check\bi(h^{-\frac23}\xi_2(x))(1+\ord(h)) \mbox{  on } \R\cap\{\re\,\xi_2(x) \leq0\},
\end{aligned}
$$
uniformly as $h\to 0_+$.
\end{proposition}
\begin{proof} The procedure is the same as for the previous proposition, but this time, setting $f( h^{-2/3}\xi_2(x)):=\xi_2'(x)^{\frac12}u(x)$, the equation $(P_2-E)u=0$ becomes,
\be
\label{eqf}
f''(t) + tf(t)=R(t)f(t),
\ee
where $R(t) =\ord (h^{4/3} (1+h^{2/3}|t|)^{-2})$. In the case of $u_{2,L}^-$, we reduce (\ref{eqf}) to the Volterra equation,
\be
\label{voltnew}
f =\check\ai+ \widetilde{\widetilde\K}[Rf],
\ee
where $\widetilde{\widetilde\K}$ is defined by,
\be
\label{defKtildedouble}
\widetilde{\widetilde\K} [f](t):=\int_{-\infty}^y K(-t,-s)f(s)ds.
\ee
In that case, we can follow the procedure of \cite{Ya} on $\{ \re\, \xi_2(x)\leq h^{-\frac23}\}$, and obtain a solution with the required asymptotics (at $-\infty$ and for $h\to 0_+$). On the other hand, on $\Gamma_+:=\{\re\, \xi_2(x)\geq 0\}$, we rewrite (\ref{voltnew}) as,
$$
f(t)=\check\ai(t) + \int_{-\infty}^0K(-t,-s)R(s)f(s)ds+\widetilde\K[Rf](t),
$$
(where $\widetilde\K$ is as in (\ref{defKKtilde})), and the asymptotics at infinity of $\check\ai$ and $\check\bi$ give,
$$
\begin{aligned}
 \widetilde\K[Rf](t) & =\ord (h^{\frac43})\sup_{\Gamma_+}|f|\int_0^{|t|} (1+|t|)^{-\frac14}(1+s)^{-\frac14}(1+h^{\frac23}s)^{-2}ds \\
 & =\ord(h(1+|t|)^{-\frac14}))\sup_{\Gamma_+}|f|.
 \end{aligned}
$$
In addition, (see also \cite{Ya}, Formula (3.18)), one has,
$$
\begin{aligned}
\int_{-\infty}^0K(-t,-s)R(s)f(s)ds & =\ord(\frac{h^{\frac43}}{(1+|t|)^{\frac14}})\sup_{\Gamma_+}|f|\int_0^\infty s^{-\frac12}(1+h^{\frac23}s)^{-2}ds\\
&  =\ord (h(1+|t|)^{-\frac14})\sup_{\Gamma_+}|f|.
\end{aligned}
$$
Thus, (\ref{voltnew}) can be solved by iteration on $\Gamma_+$, too, and there the solution satisfies,
$$
f(t)=\check\ai(t) +\ord (h(1+|t|)^{-\frac14})).
$$
The result for $u_{2,L}^-$ follows by taking $u_{2,L}^-(x):= 2(\xi_2'(x))^{-1/2}f(h^{-2/3}\xi_2(x))$.

In the case of $u_{2,L}^+$ we reduce (\ref{eqf}) to the Volterra equation,
\be
\label{volt2}
f =\check\bi+ \widetilde\K[Rf].
\ee
First working on $\{\re \, t\leq 0\}$,  the same procedure as for Proposition \ref{yaf1} leads to a solution that satisfies $f=(1+\ord(h))\check\bi$.
Next, on $\{\re \, t\geq 0\}$,  the result follows exactly as for $u_{2,L}^-$, but this time there is no need to rewrite the Volterra equation.

The asymptotic of each solution at $-\infty$ follows from that of $\ai$ and $\bi$, and the fact that $\xi_2(x)\sim c x^{2/3}$ at $+\infty$ (with $c>0$ constant).

Finally, since $\check \ai \pm i\check\bi$ do not vanish on $\R$, by setting $g:=f/(\check \ai \pm i\check\bi)$, we see (e.g. as in the proof of Proposition \ref{yaf1}) that there exist two solutions $f_\pm$ to (\ref{eqf}) on $\R_+$ satisfying $f_\pm=(1+\ord (h))(\check \ai \pm i\check\bi)$. They give rise to two solutions $v_\pm$ to $(P_2 -E)v=0$ (that are conjugated of each other for real values of $E$) and by computing their Wronskians with $u_{2,L}^\pm$, we see that they are of the form $v_\pm=\frac12a_2^-u_{2,L}\pm ia_2^+u_{2,L}^-$ with $a_2^\pm =1+\ord (h)$, so  that (\ref{asympsoolsort}) follows, too.
\end{proof}
\begin{remark}\sl
Here again, estimates on $h^{2/3}(u_{2,L}^\pm)'$ can also be derived from the construction (see Remark \ref{derivees}).
\end{remark}
\begin{remark}\sl
Near infinity, all these constructions extend to a complex sector in $x$, with the same asymptotic behaviours as on the real.
\end{remark}
\begin{remark}\sl 
\label{R+}
These constructions can also be adapted to a problem on $\R_+$ (e.g., for radial solutions of a problem in $\R^n$), with potential $V_j(r)$ ($j=1,2$) behaving like $c_j/r^\alpha$, with $c_j>0$ constant and $0\leq\alpha <2$, as $r\to 0_+$. In this case, $-\infty$ is replaced by $r=0$, and the decaying solutions at $-\infty$ become solutions vanishing at 0. Thus, in the construction of $u_{j,L}^-$ ($j=1,2$), it is enough to replace $\check\ai$ (for instance in (\ref{voltnew}))  by a linear combination $\check\ai +\alpha_j\check\bi$, where $\alpha_j\in\R$ is chosen in such a way that $\check\ai(-h^{-2/3}\xi_j(0)) +\alpha_j\check\bi(-h^{-2/3}\xi_j(0))=0$ (here $\xi_j(r)$ is the corresponding change of variable similar to (\ref{defxi2})), and to use a fundamental solution of the Airy equation vanishing at $-h^{-2/3}\xi_j(0)$ (e.g., in (\ref{defKtildedouble}), one must replace $-\infty$ by $-h^{-2/3}\xi_j(0)$). Then the proof proceeds exactly in the same way.

For potentials behaving like $c_j/r^2$ at 0 (with $c_j=c_j(h)>0$), the adaptation is even simpler since, in that case, $\xi_j(r)\to -\infty$ as $r\to 0_+$, so that the construction remains the same.
\end{remark}

\subsection{Formulae}
\label{appendix3}
$$
\begin{aligned}
 & u_{1,R}^+\sim (\xi_1')^{-\frac12}\bi (h^{-\frac23}\xi_1)\sim \bi (y-\rho)\\
 & u_{1,R}^-\sim 2(\xi_1')^{-\frac12}\ai (h^{-\frac23}\xi_1)\sim 2\ai (y-\rho)\\
 & u_{2,L}^+\sim (\xi_2')^{-\frac12}\check\bi (h^{-\frac23}\xi_2)\sim \check\bi (y+\rho)\\
 & u_{2,L}^-\sim 2(\xi_2')^{-\frac12}\check\ai (h^{-\frac23}\xi_2)\sim 2\check\ai (y+\rho)\\
  & u_{1,L}^+\sim \frac12( \cos\frac{\mathcal A}{h})u_{1,R}^- -(\sin\frac{\mathcal A}{h})u_{1,R}^+\\
 & u_{1,L}^-\sim ( \sin\frac{\mathcal A}{h})u_{1,R}^- +2(\cos\frac{\mathcal A}{h})u_{1,R}^+\\
 & u_{2,R}^+\sim \sqrt{2}e^{i\pi/4}(\frac12 u_{2,L}^-+iu_{2,L}^+)\\
 & u_{2,R}^-\sim \frac1{\sqrt{2}}e^{i\pi/4}(\frac12 u_{2,L}^--iu_{2,L}^+)
\end{aligned}
$$
$$
\widetilde\W(u_{j,L}^-, u_{j,L}^+)\sim \frac{-2}{\pi }
\quad ;\quad 
\widetilde\W(u_{j,R}^-, u_{j,R}^+) \sim \frac{2}{\pi }
$$
$$
\widetilde\W(u_{1,L}^-, u_{1,R}^-) \sim \frac{-4}{\pi }\cos\frac{\mathcal A}{h}
\quad ;\quad 
\widetilde\W(u_{2,L}^-, u_{2,R}^-)\sim \frac{i\sqrt 2}{\pi }e^{i\frac{\pi}4}
$$
$$
\widetilde\W(u_{1,L}^\pm, u_{1,R}^\mp)=\frac2{\pi}\sin\frac{\mathcal A}{h}
$$
$$
\widetilde\W(u_{2,L}^-, u_{2,R}^+) \sim \frac{-2i\sqrt 2}{\pi }e^{i\frac{\pi}4}
\quad ;\quad 
\widetilde\W(u_{2,L}^+, u_{2,R}^-)\sim \frac{1}{\pi\sqrt 2 }e^{i\frac{\pi}4}
$$

{\bf Acknowledgements} A. Martinez was partially supported by Ritsumeikan University, where the final part of this work was done, and he would like to thank the Mathematical Department of this University for his warm hospitality between April and July 2015. In addition, S. Fujii\'e and T. Watanabe thank JSPS for its financial support, and the Department of Mathematics of the University of Bologna where part of this work was done.

\end{document}